\documentclass[a4paper,twocolumn,11pt,accepted=2024-10-29]{quantumarticle}
\pdfoutput=1

\expandafter\let\csname equation*\endcsname\relax
\expandafter\let\csname endequation*\endcsname\relax

\usepackage[utf8]{inputenc}
\usepackage[english]{babel}
\usepackage[T1]{fontenc}
\usepackage{amsmath}
\usepackage{hyperref}
\usepackage{fixltx2e}
\usepackage[numbers,sort&compress]{natbib}

\usepackage[pdftex]{graphicx}
\usepackage{here}
\usepackage{amsthm}
\usepackage{amssymb, latexsym, bm, mathtools, braket, multirow, url}
\usepackage{color}

\theoremstyle{plain}

\newtheorem{thm}{Theorem}
\newtheorem{prp}[thm]{Proposition}

\newtheorem{lem}[thm]{Lemma}
\newtheorem{rem}[thm]{Remark}

\DeclareMathOperator{\tr}{\mathrm{Tr}}

\newcommand{\ol}{\overline}

\newcommand{\ep}{\epsilon}
\newcommand{\ve}{\varepsilon}

\newcommand{\lsrj}{\relbar\joinrel\twoheadrightarrow}

\newcommand{\lrto}{\leftrightarrow}

\newcommand{\midd}{\, | \,}

\newcommand{\tmix}{\text{mix}}
\newcommand{\bbC}{\mathbb{C}}
\newcommand{\bbE}{\mathbb{E}}

\newcommand{\bbR}{\mathbb{R}}
\newcommand{\bbV}{\mathbb{V}}

\newcommand{\bbZ}{\mathbb{Z}}

\newcommand{\frS}{\mathfrak{S}}

\newcommand{\clH}{\mathcal{H}}

\newcommand{\clU}{\mathcal{U}}
\newcommand{\clV}{\mathcal{V}}

\newcommand{\clX}{\mathcal{X}}
\newcommand{\sfP}{\mathsf{P}}
\newcommand{\clN}{\mathfrak{N}}

\DeclareMathOperator{\av}{av}

\DeclareMathOperator{\id}{\mathsf{id}}

\DeclareMathOperator{\SU}{SU}
\DeclareMathOperator{\Prb}{Pr}

\newcommand{\kb}[1]{\left| #1 \right> \! \left< #1 \right|}

\newcommand{\hg}[2]{{}_{#1} F_{#2}}
\newcommand{\abs}[1]{\left| #1 \right|}

\newcommand{\tket}[1]{| #1 \rangle}

\begin{document}
\title{Asymmetry activation and its relation to 
coherence under permutation operation}

\author{Masahito Hayashi}
\address{School of Data Science, The Chinese University of Hong Kong,
Shenzhen, Longgang District, Shenzhen, 518172, China}
\address{International Quantum Academy, Futian District, Shenzhen 518048, China}
\address{Graduate School of Mathematics, Nagoya University, Furo-cho, Chikusa-ku, Nagoya, 464-8602, Japan}
\orcid{0000-0003-3104-1000}
\email{hmasahito@cuhk.edu.cn, masahito@math.nagoya-u.ac.jp}
\maketitle

\begin{abstract}
A Dicke state and its decohered state are invariant for permutation.
However, when another qubits state to each of them is attached, the whole state is not invariant for permutation, and has a certain asymmetry for permutation.
The amount of asymmetry can be measured by the number of distinguishable states under the group action or the mutual information.
Generally, the amount of asymmetry of the whole state 
is larger than the amount of asymmetry of the added state.
That is, the asymmetry activation happens in this case.
This paper investigates the amount of the asymmetry activation
under Dicke states.
To address the asymmetry activation asymptotically, we introduce a new type of central limit theorem by using several formulas on hypergeometric functions.
We reveal that the amounts of the asymmetry and the asymmetry activation
with a Dicke state have a strictly larger order than those with the decohered state in a specific type of the limit.
\end{abstract}

\section{Introduction}\label{s:intro}
Asymmetry is a root of the variety of various physical phenomena.
Consider a perfectly symmetric state. Such a state does not reflect any changes so that 
it can realize no variety.
To realize rich variety, the state needs to have sufficient asymmetry.
In physics, people consider that symmetry breaking is needed for our world.
Therefore, to explain various phenomena, 
many researchers introduced spontaneous symmetry breaking, which plays 
an important role in various areas in physics,
nuclear star, superfluidity, superconductivity, cold atoms, Higgs boson, Nambu-goldstone boson.
When the degree of symmetry breaking is large, it realizes rich phenomena.
Since the amount of symmetry breaking can be restated as
the amount of asymmetry, it is a central topic to study the amount of asymmetry.

Fortunately, quantum information theory has strong tools to evaluate the amount of asymmetry \cite{VAW,Marvian}.
The variety generated by group symmetry under a given state
can be measured by the number of distinguishable states generated by group action to describe the 
symmetry.
This number is evaluated by 
the difference between 
the von Neumann entropy of the averaged state with respect to the group action
and the von Neumann entropy of the original state.
This fact can be shown by a simple application of quantum channel coding theorem \cite{Holevo,SW}
to the family of the states translated by the group action 
although this idea can be backed to various studies \cite{Hiroshima,H15,Korzekwa}.
Therefore, the difference between the von Neumann entropies of these two states
can be considered as the amount of the asymmetry of the given state.
Now, we focus on the permutation symmetry on qubits
similar to the recent paper \cite{ICM}.
When we apply the permutation of the state $\ket{1^M \, 0^N}$,
the averaged state $\rho_{mix,N+M,M}$ is permutation-invariant, and has no asymmetry.
However, it still works as a resource for asymmetry 
when it is attached to another qubits state $\ket{1^l \, 0^{k-l}}$.
When we have only the state $\ket{1^l \, 0^{k-l}}$,
the amount of asymmetry is measured by the von Neumann entropy $S(\rho_{mix,k,l})$
of $\rho_{mix,k,l}$.
When the state $\rho_{mix,N+M,M}$ is attached,
the amount of asymmetry is $S(\rho_{mix,k+N+M,l+M})-S(\rho_{mix,N+M,M})$,
which is larger than the asymmetry $S(\rho_{mix,k,l})$ of the added state.
In this paper, we focus on the increase of the amount of asymmetry by adding an invariant state,
and call this phenomena an asymmetry activation.
In fact, an asymmetry activation with various groups often
happens behind of symmetry breaking. 
In this paper, we study the asymmetry activation
under quantum information-theoretical objects
because no systematic research studied the asymmetry activation
in the area of quantum information.

Now, we discuss how quantum theory can improve 
the asymmetry activation
because the above state is classical.
To address this problem, we focus on the coherence, which is a key concept in
quantum theory \cite{Streltsov}.
That is, we consider 
a permutation-invariant state with great coherence.
We focus on the Dicke state $ \ket{\Xi_{N+M,M}}$ with the same weight as
the classical state $\rho_{mix,N+M,M}$.
This state has the largest coherence among 
invariant states with the same weight as the state $\rho_{mix,N+M,M}$
because any invariant state with the same weight has the same diagonal elements as 
the state $\rho_{mix,N+M,M}$.
A Dicke state has been playing an important role in 
the calculation of entanglement measures \cite{WG,HMMOV,WEGM,W,ZCH}, 
quantum communication, and quantum networking \cite{KSTSW,PCTPWKZ}. 
The experiments \cite{KSTSW,PCTPWKZ} are motivated by the fact the overlap
of symmetric Dicke states with biseparable states is close to
1/2 for large $N$ \cite{Toth}.
Multipartite entanglement has been detected in an ensemble of thousand of
atoms realizing a Dicke state \cite{LPVA}.
Some typical Dicke states have been realized in trapped atomic ions \cite{HCTW}. 
Recently, the multi-qubit Dicke state with half-excitations
has been employed to implement a scalable quantum search based on
Grover's algorithm by using adiabatic techniques \cite{IILV}. 
These studies show the importance of Dicke states

In particular, the Dicke state can be used for quantum metrology, which has been tested experimentally in cold gases of thousands of atoms \cite{LSKP,HGHB}.
The Dicke states are optimal in linear interferometers, in the sense,
that they saturate the inequality
$F_Q[J_x,\rho]+F_Q[J_y,\rho]+F_Q[J_z,\rho] \le N(N+2)/4$,
which are saturated also by GHZ states \cite{HLKW,Toth2}. Here, $F_Q$ is the quantum Fisher
information corresponding to a unitary dynamics.
Also the reference \cite{HO} studied the estimation of the same direction with Dicke states.
These studies consider the information only for the direction of the rotation by the special unitary group.
They do not focus on the direction of permutation because 
the Dicke state is invariant for permutation.
It is a completely new idea to use the information of the direction of permutation by modifying 
the Dicke state. 
To execute this new idea,
using non-negative integers $n,m,k,l$ as $M:=m-l$, $N:=n-m-k+l$, we consider
the asymmetry of the following state.
\begin{align}\label{eq:intro:Xi}
 \ket{\Xi_{n,m|k,l}} := \ket{1^l \, 0^{k-l}} \otimes \ket{\Xi_{N+M,M}}
 \in (\bbC^2)^{\otimes n}.
\end{align}
The amount of asymmetry is the von Neumann entropy
$S(\av_\pi (
\ket{\Xi_{n,m|k,l}}))$,
where $\av_\pi$ expresses the averaged state with respect to the permutation.

However, the calculation of this von Neumann entropy is not simple
while $S(\rho_{mix,k,l})$ is calculated to $\log \binom{k}{l}$, which is approximated to
$l h(k/l)$, where $h(x):=-x\log x-(1-x)\log (1-x)$ is the binary entropy.
This calculation is closely related to Schur-Weyl duality, one of the key 
structure in the representation theory.
When we make a measurement given by the decomposition defined by 
Schur-Weyl duality for the system whose state is
$\ket{\Xi_{n,m|k,l}}$,
we obtain the distribution of this measurement outcome.
This distribution takes a key role in this calculation,
The paper \cite{HHY1} 
discovered notable formulas for this distribution
by using Hahn and Racah polynomials, which are 
$\hg{3}{2}$- and $\hg{4}{3}$-hypergeometric orthogonal polynomials, respectively.
It is possible to numerically calculate
this von Neumann entropy by using these tools.

Even though its numerical calculation is possible, 
it is quite difficult to understand its behavior.
To grasp its trend, we consider two kinds of asymptotics, i.e., Type I and Type II limits.
Type I limit assumes that $k$ and $l$ are fixed and $m$ is linear for $n$, 
i.e., $m=\xi n$ with fixed ratio $\xi$, under $n \to \infty$.
Type II limit does that $k,l$ and $m$ are linear for $n$ under $n \to \infty$.
Type I limit corresponds to the situation of the law of small numbers,
and Type II limit corresponds to the situation of the central limit theorem.
For the relation of 
Type I and Type II limits
with the law of small numbers and 
the central limit theorem, see Appendix \ref{S-binomial}.
Based on these formulas, 
the paper \cite{HHY2} derived the asymptotic approximation of the distribution in the limit $n \to \infty$.

In Type I limit, we show that
the von Neumann entropy $S\bigl(\av_\pi[\Xi_{n,m|k,l}]\bigr)$ 
behaves  as
\begin{align}
\begin{split}
&S\bigl(\av_\pi[\Xi_{n,m|k,l}]\bigr)\\
= &((k-l)\xi+l(1-\xi)) \log n +O(1).
\end{split}
\label{eq:intro:S-asymp}
\end{align}
The same scaling $\log n$ appears in the bosonic coherent channel when 
the total power is fixed and the number $n$ of modes increases \cite{Ha10-1}.
Figure \ref{fig:intro:vN} shows an example of this approximation.
When we use the decohered state $\rho_{mix,N+M,M}$,
the amount of asymmetry is calculated as
\begin{align}
\begin{split}
&S(\rho_{mix,k+N+M,l+M})-S(\rho_{mix,N+M,M}) \\
=& k h( \xi)- h'( \xi)(\xi k-l)
+O(1) .\label{NHI1B}
\end{split}
\end{align}
The amount of asymmetry in \eqref{NHI1B}
is a constant for $n$
while the case with Dicke state $  \ket{\Xi_{N+M,M}}$ 
has the order $\log n $ as \eqref{eq:intro:S-asymp}.
In particular, 
in this scenario, while the size of the added system is finite,
the amount of activated asymmetry goes to infinity.
Therefore, 
the Dicke state $  \ket{\Xi_{N+M,M}}$ significantly improves
the amount of asymmetry over the decohered state
$\rho_{mix,N+M,M}$ under Type I limit.

In Type II limit, i.e., the limit $n \to \infty$ with $k,l$ and $m$ linear for $n$,
we need to treat many ratios and parameters,
which are summarized in Table \ref{tab:intro:IIprm}.
While 
the fixed ratios 1 and 2 are natural,
the fixed ratio 3 is useful for our purpose.

\begin{table}[htbp]
\begin{tabular}{l|l}
\hline
 fixed ratio 
 & $\alpha=\frac{l  }{n}$,~ $\beta=\frac{m-l}{n}$ \\
 set 1 
&   $\gamma=\frac{k-l}{n}$,~ $\delta=\frac{n-m-k+l}{n}$ \\
\hline
 fixed ratio 
 & $\alpha=\frac{l  }{n}$,~ $\xi=\frac{m}{n}$ \\
 set 2 &
   $\kappa=\frac{k  }{n}$ \\
   [1ex] 
 fixed ratio 
 & $\beta=\frac{m-l}{n}$,~  $\delta=\frac{n-m-k+l}{n}$\\
 set 3 &
   $\xi=\frac{m}{n}$  \\
   \hline 
   \hline 
 limit pdf 
 &$\mu := \frac{1-\sqrt{D}}{2}$ ,~
  $\sigma^2  := \frac{(1-\beta-\delta)\beta\delta}{D} $ \\
 parameters 
 & $D := 4 \beta \delta+ (2\xi-1)^2$\\
   \hline 
 \end{tabular}
\caption{Ratios and parameters for Type II limit}
\label{tab:intro:IIprm}
\end{table}

Then, we show that
\begin{align}
 S(\av_{\pi}[\Xi_{n,m|k,l}]) = n h(\mu)+o(n)\label{DSAP}.
\end{align}
Since the amount of asymmetry of the added system is 
$\log \binom{k}{l}$ and is approximated to
$n (\alpha+\gamma) h(\frac{\alpha}{\alpha+\gamma})+o(n)$, the amount of 
the asymmetry activation is approximated to
$n(h(\mu)-(\alpha+\gamma) h(\frac{\alpha}{\alpha+\gamma})) +o(n)$.

When we use the decohered state $\rho_{mix,N+M,M}$,
the amount of asymmetry is calculated as
\begin{align}
\begin{split}
&S(\rho_{mix,k+N+M,l+M})-S(\rho_{mix,N+M,M})\\
=&
n(h(\xi)-(\beta+\delta)h(\frac{\beta}{\beta+\delta}))
+o(n).
\end{split}
\label{NHI2Y}
\end{align}
Hence,
the amount of the asymmetry activation is
approximated to
$n(h(\xi)-(\beta+\delta)h(\frac{\beta}{\beta+\delta}))
-(\alpha+\gamma) h(\frac{\alpha}{\alpha+\gamma})) +o(n)$.
For this derivation, we employ a kind of central limit theorem
obtained in \cite{HHY2}.
Since the inequality $h(\mu)\ge h(\xi)-(\beta+\delta) h(\frac{\beta}{\beta+\delta})$
holds as shown later,
the Dicke state $  \ket{\Xi_{N+M,M}}$ significantly improves
the amount of asymmetry over the decohered state
$\rho_{mix,N+M,M}$ even under Type II limit.
In addition, similar characterization is also possible from the viewpoint of the number of distinguished states.

When the Dicke state $ \ket{\Xi_{N+M,M}}$
is decohered, the state changes to the classical state $\rho_{mix,N+M,M}$.
Since both states are invariant for permutation,
their difference is the existence/non-existence of the coherence.
Therefore, the difference between \eqref{eq:intro:S-asymp} and \eqref{NHI1B}
and the difference between \eqref{DSAP} and \eqref{NHI2Y}
can be considered as the improvement of 
asymmetry by using the coherence.
In addition, the formula \eqref{DSAP} has a very simple form 
and is characterized by the binary entropy of the quantity $\mu$.
Since the quantity $\mu$ has not appeared in physics, 
the formula \eqref{DSAP} means that
the quantity $\mu$ is a newly discovered quantity
that has a physical role related to asymmetry.
Therefore, we can expect that 
the new quantity $\mu$ takes certain roles in other topics.

The remaining part of this paper is organized as follows.
As a preparation, Section \ref{S15} discusses the general theory for the amount of asymmetry,
and prepares several formulas for the number of distinguished states.
Also, it explains the general theory for asymmetry activation.
Section \ref{S2-1} gives our problem setting, and introduces
a key distribution to determine the amount of asymmetry.
Section \ref{S2-15} discusses the amount of asymmetry for the decohered state.
Section \ref{S2-2} reviews the asymptotic behavior of the distribution related to
Schur-Weyl duality, which was obtained in the paper \cite{HHY2}.
Section \ref{S2-3} derives the asymptotic behavior of the von Neumann entropy of the average state, 
which expresses the asymptotic behavior.
Section \ref{S2-4} considers the  asymptotic behavior of 
the number of distinguishable states.
Section \ref{S6} gives another example for the asymmetry activation.
Section \ref{S5} gives the conclusion and discussion.

Appendix \ref{s:SS} gives the proofs of the statements given in Section \ref{S2-1}.
Appendix \ref{NVR} gives the derivation of the statement given in Section \ref{S2-15}.
Appendix \ref{NVR2} gives the derivation of the statements given in Sections \ref{S2-3} and \ref{S2-4}.
Appendix \ref{ss:II:sp} shows 
an asymptotic formula for the von Neumann entropy 
presented in Section \ref{S2-3}.

\section{General theory}\label{S15B}
\subsection{General theory for asymmetry}\label{S15}
First, we discuss a general theory for asymmetry when
a unitary representation $f$ of a finite group $G$ is given 
on $\clH$.
Since the invariance of a state $\rho$ under the $G$-action $f$ 
is expressed as $f(g) \rho f(g)^\dagger= \kb{\psi}$ for any $g \in G$,
the asymmetry of $\rho $ is considered 
as the degree of change caused by the $G$-action.
When a state $f(g) \rho  f(g)^\dagger$ is generated with equal probability $\frac{1}{\abs{G}}$,
the averaged state with respect to the representation $f$
is given as
\begin{align}
 \av_f[\rho ] := \sum_{g \in G} \frac{1}{\abs{G}} f(g) \rho  f(g)^\dagger.
\end{align}
In this case, the map $g \mapsto f(g) \rho  f(g)^\dagger$
is considered as a classical-quantum channel, and
the mutual information is given as
\begin{align}
\begin{split}
&S(\av_f[\rho ])-
\sum_{g \in G} \frac{1}{\abs{G}} 
S(f(g) \rho f(g)^\dagger) \\
=&
S(\av_f[\rho ])-S( \rho),\label{BIE}
\end{split}
\end{align}
where
the von Neumann entropy $S(\rho)$ is defined as
$-\tr \rho \log \rho$.
If we use the relative entropy $D(\rho\|\sigma)=\tr \rho(\log \rho-\sigma) $, 
the mutual information is 
rewritten as
\begin{align}
\sum_{g \in G} \frac{1}{\abs{G}} 
D(f(g) \rho f(g)^\dagger \| \av_f[\rho ])
=
D(\rho  \| \av_f[\rho ])
\end{align}
because $D(f(g) \rho f(g)^\dagger \| \av_f[\rho ])
=D(\rho  \| \av_f[\rho ])$.
That is, the asymmetry is measured by
$S(\av_f[\rho ])-S( \rho)=D(\rho  \| \av_f[\rho ])$.
Therefore, the mutual information \eqref{BIE}
can be considered as the degree of asymmetry of 
$\rho$ under the representation $f$.
This quantity was also introduced in the context of the resource theory of asymmetry
\cite{Marvian} and other context \cite{VAW}.
In particular, when $\rho$ is a pure state $\kb{\psi}$,
it is simplified to 
\begin{align}\label{eq:intro:vN}
 S(\av_f[\psi]) := -\tr_{\clH}\bigl(\av_f[\psi] \log \av_f[\psi]\bigr).
\end{align}
Further, when a channel $\Gamma$ satisfies the covariance condition:
\begin{align}
\Gamma( f(g) \rho f(g)^\dagger)=
f(g) \Gamma( \rho) f(g)^\dagger,\label{NFG2}
\end{align}
the information processing inequality for $\Gamma$ implies
\begin{align}
&S(\av_f[\rho ])-S( \rho) =
D(\rho  \| \av_f[\rho ])\nonumber\\
\ge & D(\Gamma (\rho)  \| \Gamma(\av_f[\rho ])) 
= D(\Gamma (\rho)  \| \av_f[\Gamma(\rho) ])
\nonumber\\
=&S(\av_f[\Gamma (\rho) ])-S( \Gamma (\rho)) .\label{NFG}
\end{align}
This fact shows that the application of a covariant channel
decreases the amount of asymmetry.

In fact, to connect the mutual information \eqref{BIE}
and the distinguishability, we need to consider the 
repetitive use of the classical-quantum channel defined
by $\{f(g) \rho f(g)^\dagger \}_{g \in G}$.
When our interest is the case when the above classical-quantum channel is used only once,
we need to prepare another measure.
In this case, we evaluate the number $M(\rho,\ep)$ of distinguished states among 
$\{f(g) \rho f(g)^\dagger \}_{g \in G}$ with error probability $\ep$.

Assume that the state $\rho$ is 
a constant times of a projection to a $v$-dimensional subspace
and the set
$\{f(g) \rho f(g)^\dagger \}_{g \in G}$ 
is composed of orthogonal states.
In this case
$\av_{f}[\rho]$ is a constant times of a projection to a $w$-dimensional subspace.
In this case, the number $M(\rho,0)$
of distinguishable state is $w/v$.
Also, $S(\av_{f}[\rho])- S(\rho)$ equals $\log w/v$.
When we allow error probability $1-\frac{1}{j}$ with a positive integer $j$,
the number 
of distinguishable inputs is $jw/v$.
This is because it is allowed that $j$ informations are represented to the same state.
Hence, we have the formula
\begin{align}
\log M(\rho, 1-\frac{1}{j})=
\log j + \log \frac{w}{v}.\label{NXU}
\end{align}

Next, we prepare a general formula for $M(\rho,\ep)$ for a general state $\rho$
when 
the set of states 
$\{f(g) \rho f(g)^\dagger \}_{g \in G}$ is given.
We prepare the following notation.
Given a Hermitian matrix $A$, the symbols $\{ A \ge 0\}$ and $\{ A > 0\}$
express the projections to the eigenspace of $A$ corresponding to non-negative
and strictly positive eigenvalues, respectively. 
For two Hermitian matrices $A$ and $B$, 
the symbol $\{A \ge B\}$ denotes the projection $\{A-B\ge 0\}$.
Other projections such as $\{A>B\}$ and $\{A \le B\}$ are similarly defined.
Then, we consider the information spectrum relative entropy $D_s^\delta$,
which is given by \cite{hayashi2003general}:
\begin{align}\label{eq:spectrum_rel_ent}
 D_s^\delta(\rho\|\rho')
 := \max \bigl\{\lambda\ \mid \tr \rho \{\rho\leq e^\lambda \rho'\} 
               \le \delta \bigr\}.
\end{align}

Since $\av_f[\rho]$ is invariant for $f$,
applying \cite[Lemma 4]{hayashi2003general} to the case $\sigma=\av_f[\rho]$, we have
\begin{align}\label{converse}
 \log M(\rho,\ep)
 \leq D_s^{\ep+\delta}( \rho \| \av_f[\rho]) + \log \frac{1}{\delta}.
\end{align}
The paper \cite[Theorem 2]{Korzekwa} derived the following evaluation:
\begin{align}\label{direct}
 \log M(\rho,\ep) \ge
 D_s^{\ep-\delta}(  \rho\| \av_f[\rho] )) - \log \frac{1}{\delta} ,
\end{align}
which is slightly better than \cite[Lemma 3]{hayashi2003general} in this special case.
We define 
\begin{align}
&   H_s^{\ep}(\rho)
:= \max \bigl\{\lambda \mid \tr \rho \{ - \log \rho \le \lambda \} \le \ep \bigr\} \nonumber\\
 =& \max \bigl\{\lambda \mid \tr \rho \{   \rho \ge e^{-\lambda} \} \le \ep \bigr\}.
\end{align}
When $\rho$ is a pure state, as shown Theorem \ref{TH2} in \ref{s:SS}, we have
\begin{align}
 M(\rho,\ep) & \ge H_s^{\ep-\delta_1-\delta_2}(\av_f[\rho])
                   - \log \frac{1}{\delta_1 \delta_2},   \label{EE9} \\
 M(\rho,\ep) & \le H_s^{\ep+\delta_1+2\delta_2}(\av_f[\rho]) 
                   + \log \frac{1}{\delta_1 \delta_2^2}. \label{EE10}
\end{align}

\subsection{General theory for asymmetry activation}\label{S15B}
The asymmetry activation is generally formulated as follows.
We consider two quantum systems ${\cal H}_1$ and 
${\cal H}_2$, and two unitary representation $f_1$ and $f_2$ of compact groups
$G_1$ and $G_2$ on ${\cal H}_1$ and ${\cal H}_2$, respectively.
Then, we consider a larger compact group $G_3$ that contains $G_1\times G_2$
and its unitary representation $f_3$ on ${\cal H}_1\otimes {\cal H}_2$, which satisfies
$f_3|_{G_1 \times G_2}= f_1 \otimes f_2$.
We choose an invariant state $\rho_1$ on ${\cal H}_1$
and a state $\rho_2$ on ${\cal H}_2$.
Then, the amount of asymmetry of the original state is 
$D(\rho_1\|\av_{f_1}[\rho_1])+D(\rho_2\|\av_{f_2}[\rho_2])
=D(\rho_2\|\av_{f_2}[\rho_2])$.
In contrast, the amount of asymmetry of the whole state is
$D(\rho_1\otimes \rho_2\|\av_{f_3}[\rho_1\otimes \rho_2])$.
Hence, the amount of the asymmetry activation is calculated as
\begin{align}
&D(\rho_1\otimes \rho_2\|\av_{f_3}[\rho_1\otimes \rho_2])-
D(\rho_2\|\av_{f_2}[\rho_2]) \notag\\
=&-S(\rho_1)-S( \rho_2)+S(av_{f_3}[\rho_1\otimes \rho_2])\notag\\
&+S(\rho_2)-S(\av_{f_2}[\rho_2]) \notag\\
=&-S(\rho_1)-S(\av_{f_2}[\rho_2])+S(av_{f_3}[\rho_1\otimes \rho_2])\notag\\
=&D(\rho_1\otimes \av_{f_2}[\rho_2]\|\av_{f_3}[\rho_1\otimes \rho_2]).
\end{align}

Now, we choose a TP-CP map $\Gamma_1$ on ${\cal H}_1$
and a TP-CP map $\Gamma_2$ on ${\cal H}_2$.
When the covariance condition
\begin{align}
\Gamma_1\otimes \Gamma_2 ( f_3(g) \rho f_3(g)^\dagger)=
f_3(g) \Gamma_1\otimes \Gamma_2( \rho) f_3(g)^\dagger \label{NFGR}
\end{align}
holds for $g \in G_3$, 
the information processing inequality for $\Gamma_1\otimes \Gamma_2 $ implies
\begin{align}
&D(\rho_1\otimes \av_{f_2}[\rho_2]\|\av_{f_3}[\rho_1\otimes \rho_2]) \notag\\
\ge 
&D(\Gamma_1[\rho_1]\otimes \Gamma_2[\av_{f_2}[\rho_2]]\|
(\Gamma_1\otimes \Gamma_2)[\av_{f_3}[\rho_1\otimes \rho_2]]) \notag\\
=&
D(\Gamma[\rho_1]\otimes \av_{f_2}[\Gamma_2[\rho_2]]\|
\av_{f_3} [\Gamma_1[\rho_1]\otimes \Gamma_2[\rho_2]]).
\label{NM67}
\end{align}
When the decoherence operator satisfies the above condition \eqref{NFGR},
the decoherence reduces the amount of the asymmetry activation.

\section{Formulation}\label{S2-1}
This paper investigates how an invariant state activates the amount of asymmetry.
To address this question, we focus on 
the $n$-fold tensor product 
$\clH=(\bbC^2)^{\otimes n}$ of the qubit system $\bbC^2=\bbC \ket{0} \oplus \bbC \ket{1}$,
and consider it as the representation space of the permutation group $G=\frS_n$ 
acting by permutation of tensor factors as 
\begin{align}\label{eq:intro:Sn-act}
\pi(g)(\ket{i_1 \dotsm i_n}) := \ket{i_{g^{-1}(1)} \dotsm i_{g^{-1}(n)}}, 
 \quad g \in \frS_n.
\end{align}
Then, we apply the general theory given in Subsection \ref{S15B}
when
${\cal H}_1=(\bbC^2)^{\otimes N+M}$,
${\cal H}_2=(\bbC^2)^{\otimes k}$,
$G_1=\frS_{N+M}$,
$G_2=\frS_{k}$,
$G_3=\frS_{N+M+k}$
and 
$\rho_2=\ket{1^l \, 0^{k-l}} \in (\bbC^2)^{\otimes k}$ is defined as
\begin{align}\label{eq:intro:1^a}
   \ket{1^l \, 0^{k-l}} 
:= \tket{\overbrace{1 \dotsm 1}^{l} \, \overbrace{0 \dotsm 0}^{k-l}}
 =  \ket{1}^{\otimes l} \otimes \ket{0}^{\otimes (k-l)} 
\end{align}
while the state 
$\ket{1^l \, 0^{k-l}}$ is asymmetric with respect to the permutation.

As one possible choice of $\rho_1$, we choose 
a Dicke state on $(\bbC^2)^{\otimes (N+M)}$, which 
is a typical invariant state because 
it is given as the permutation invariant state 
with fixed weights $N$ and $M$ as
\begin{align}
& \ket{\Xi_{N+M,M}} \nonumber\\
:=& \tbinom{N+M}{M}^{-1/2} 
 \bigl(\ket{1^M \, 0^N} + \text{permuted terms}\bigr).
\label{eq:0:Dicke}
\end{align}
Remember that the state $ \ket{\Xi_{n,m|k,l}}$ is defined as 
\begin{align}
 \ket{\Xi_{n,m|k,l}} = \ket{1^l \, 0^{k-l}} \otimes \ket{\Xi_{N+M,M}}
 \in (\bbC^2)^{\otimes n}.
\end{align}
In contrast, the decohered state of the Dicke state
is given as
\begin{align}\label{symB}
\begin{split}
   &\rho_{mix,N+M,M} \\
   :=& \tbinom{N+M}{M}^{-1}
 \bigl(\kb{1^M \, 0^N} + \text{permuted terms}\bigr).
\end{split}
\end{align}

To study the asymmetry, we focus on the $\SU(2)$-$\frS_n$ Schur-Weyl duality
on the tensor product system $\clH=(\bbC^2)^{\otimes n}$ depicted as
\begin{align}\label{eq:intro:SCS}
 \SU(2) \curvearrowright (\bbC^2)^{\otimes n} \curvearrowleft \frS_n.
\end{align}
The Schur-Weyl duality claims that we have the following decomposition of 
$(\bbC^2)^{\otimes n}$ into irreducible representations (irreps for short):
\begin{align}\label{eq:intro:SW}
 (\bbC^2)^{\otimes n} 
 = \bigoplus_{x=0}^{{n/2}} \clU_{(n-x,x)} \boxtimes \clV_{(n-x, x)}.
\end{align} 
Here $\boxtimes$ denotes the tensor product of linear spaces, equipped with 
$\SU(2)$-action on the left and $\frS_n$-action on the right factor,
$\clV_{(n-x, x)}$ denotes the $\frS_n$-irrep corresponding to the partition $(n-x,x)$, 
and $\clU_{(n-x,x)}$ denotes the highest weight $\SU(2)$-irrep of dimension $n-2x+1$.
Using the hook formula (see \cite[I.5, Example 2; I.7, (7.6)]{M} for example),
we can compute the dimension of $\clV_{(n-x, x)}$ as 
\begin{align}\label{eq:hook}
\begin{split}
& \dim \clV_{(n-x,x)} = \frac{k!}{\prod_{\square \in (n-x,x)} h(\square)} \\
 =& 
 \frac{n-2x+1}{n-x+1} \binom{n}{x} = \binom{n}{x} - \binom{n}{x-1},
\end{split}
\end{align}
where $\square \in (n-x,x)$ denotes a box in the Young diagram 
corresponding to the partition $(n-x,x)$.
Using this formula, we have the completely mixed state 
$\rho_{\tmix, \clV_{(n-x,x)}} = \frac{1}{\dim \clV_{(n-x,x)}} \id_{\clV_{(n-x,x)}}$.
The highest weight $\SU(2)$-irrep of dimension $2j+1$ has the standard basis 
$ \{ \ket{j,m} \mid m=-j,-j+1, \dotsc, j\}$, where
$j$ and $m$ express the total and z-direction angular momenta, respectively.

In our analysis, the projector $\sfP_{(n-x,x)}$ onto the isotypical component 
in the decomposition \eqref{eq:intro:SW}
plays a key role, and is defined as
\begin{align}\label{eq:intro:proj}
 \sfP_{(n-x,x)}\colon 
 (\bbC^2)^{\otimes n} \lsrj \clU_{(n-x,x)} \boxtimes \clV_{(n-x,x)}.
\end{align}
As the set of the projectors 
$\{\sfP_{(n-x,x)}\}_{x}$ forms a projective measurement,
we define the probability $p(x \midd n,m,k,l)$ as
\begin{align}\label{eq:intro:p(x)}
  p(x \midd n,m,k,l) := 
 \bra{\Xi_{n,m|k,l}} \sfP_{(n-x,x)} \ket{\Xi_{n,m|k,l}},
\end{align}
and denote the corresponding probability distribution by $P_{n,m,k,l}$.
Since we can switch $\ket{0} \lrto \ket{1}$ in the system, 
the probability mass function (pmf) enjoys the system symmetry
\begin{align}\label{eq:sym}
 p(x \midd n,m,k,l) = p(x \midd n,n-m,k,k-l), 
\end{align}
Since the distribution $P_{n,m,k,l}$ describes 
the averaged state $\av_\pi[\Xi_{n,m|k,l}]$ as 
\begin{align}
\begin{split}
& \av_\pi[\Xi_{n,m|k,l}] \\
=& \sum_{x=0}^{{n/2}} p(x \midd n,m,k,l) 
 \kb{\tfrac{n}{2}-x , \tfrac{n-2m}{2}} \\
& \boxtimes \rho_{\tmix, \clV_{(n-x,x)}},
\label{ALO}
\end{split}
\end{align}
our analysis is concentrated in the distribution
$P_{n,m,k,l}$.
We also denote by $X$ a discrete random variable distributed by 
the probability mass function (pmf for short) $p(x \midd n,m,k,l)$.
We choose integers $n,m,k$ and $l$ from the following set;
\begin{align}\label{eq:intro:clN}
\begin{split}
 &\clN \\
 :=& 
 \Big\{(n,m,k,l) \in \bbZ_{\ge 0}^4 \Big |
  \begin{array}{l}
  m, k \le n, \\
  \ m+k-n \le l \le m \wedge k
  \end{array}
    \Big\}.
\end{split}
\end{align}

Then, the von Neumann entropy is given by 
\begin{align}\label{eq:intro:vNav}
\begin{split}
 &S\bigl(\av_\pi[\Xi_{n,m|k,l}]\bigr) \\
 = &
 \sum_x  p(x \midd n,m,k,l) \\
 &\cdot\bigl(\log \dim \clV_{(n-x,x)}-\log p(x\midd n,m,k,l)\bigr),
\end{split}
\end{align}
which expresses the amount of asymmetry of $\ket{\Xi_{n,m|k,l}}$.
Furthermore, we can also consider the number of approximately orthogonal elements among 
the states $\{ \pi(g) \ket{\Xi_{n,m|k,l}} \}_{g \in \frS_n}$,
which is characterized by a slightly different quantity.
As explained in Section \ref{S2-4}, the number of distinguishable elements
with error probability $\ve>0$ is approximated by the following value:
\begin{align}\label{eq:intro:XAL}
\begin{split}
&\ol{H}_s^{\ve}(\av_\pi[\Xi_{n,m|k,l}])  \\
&:= \max\biggl\{ \lambda \bigg| 
P_{n,m,k,l}[\{x \mid \log \dim \clV_{(n-x,x)} \\
&        \hspace{10ex}        -\log p(x \midd n,m,k,l) \le \lambda\}] \le \ep \biggr\}.
\end{split}
\end{align}

\section{Asymmetry with decohered symmetric state}\label{S2-15}
First, we discuss the case with the decohered state 
$\rho_{mix,N+M,M}$.
Then, we have
\begin{align}
& S(\av_\pi[\kb{1^l \, 0^{k-l}}\otimes \rho_{mix,N+M,M}] )\nonumber \\
&-S(\kb{1^l \, 0^{k-l}}\otimes \rho_{mix,N+M,M} )\nonumber \\
=&S(\rho_{mix,k+N+M,l+M} )
-S(\rho_{mix,N+M,M} )\nonumber \\
=&\log \binom{l+M}{k+N+M}-\log \binom{M}{N+M}\nonumber \\
=&\log \binom{m}{n}-\log \binom{m-l}{n-k}.
\label{NBT}
\end{align}

To consider the relation with $ S\bigl(\av_\pi[\Xi_{n,m|k,l}]\bigr) $,
we consider the pinching map ${\cal E}$ defined as
\begin{align}
{\cal E}(\rho):= \sum_{\vec{x} \in \{0,1\}^n}
\kb{\vec{x}} \rho \kb{\vec{x}},
\end{align}
which expresses the decoherence process.
Since the pinching map ${\cal E}$ satisfies the covariance condition \eqref{NFG2}
with respect to permutation,
the relation \eqref{NFG} yields the inequality
\begin{align}
& S\bigl(\av_\pi[\Xi_{n,m|k,l}]\bigr) 
\nonumber \\
\ge &S(\rho_{mix,k+N+M,l+M} )
-S(\rho_{mix,N+M,M} ). \label{NMT}
\end{align}
This inequality also follows from 
the inequality \eqref{NM67} for the amount of asymmetry activation.
Therefore, the Dicke state realizes a larger amount of asymmetry than 
the decohered state.
However, this inequality does not explain how much 
the Dicke state improves it over the decohered state.
To answer this question, we need to consider the asymptotic behavior.

Here, we consider the number 
$M(\kb{1^l \, 0^{k-l}}\otimes \rho_{mix,N+M,M},\epsilon)$ as well.
Counting the possible combinatorics, we find that
the number of discriminated states is calculated as
$M(\kb{1^l \, 0^{k-l}}\otimes \rho_{mix,N+M,M},0)
= \binom{l+M}{k+N+M}/ \binom{M}{N+M}$.
Since the set
$\{\pi(g) \kb{1^l \, 0^{k-l}}\otimes \rho_{mix,N+M,M}
 \pi(g)^\dagger \}_{g \in \frS_n}$ 
is composed of orthogonal states, due to \eqref{NXU},
the difference between 
$M(\kb{1^l \, 0^{k-l}}\otimes \rho_{mix,N+M,M},0)$
and $M(\kb{1^l \, 0^{k-l}}\otimes \rho_{mix,N+M,M},\epsilon)$
with finite error probability $\ep$ is a constant.
Hence, we discuss the asymptotics of 
the value given in \eqref{NBT}.

In Type I limit, we have
\begin{align}
&\log \binom{m}{n}-\log \binom{m-l}{n-k}\nonumber \\
=&k h( \xi)- h'( \xi)(\xi k-l)
+o(1) .\label{NHI1}
\end{align}
That is, the amount of asymmetry is a constant for $n$ in this case.
In Type II limit, we have
\begin{align}
&\log \binom{m}{n}-\log \binom{m-l}{n-k}\nonumber \\
=&
n(h(\xi)-(\beta+\delta) h(\frac{\beta}{\beta+\delta}))
+\frac{1}{2}\log (\beta+\delta) \nonumber \\
&-\tfrac{1}{2}\log \bigl( 2\pi \xi(1-\xi)\bigr)  \nonumber\\
&+\tfrac{1}{2}\log \bigl( 2\pi \frac{\beta}{\beta+\delta}(1-\frac{\beta}{\beta+\delta})\bigr)+o(1) .\label{NHI2}
\end{align}
The relations \eqref{NHI1} and \eqref{NHI2} are shown in 
Appendix \ref{NVR}

\section{Asymptotics analysis on distribution $P_{n,m,k,l}$}
\label{S2-2}
To derive the asymptotic behaviors of  
$S\bigl(\av_\pi[\Xi_{n,m|k,l}]\bigr)$ and 
$\log M({\Xi_{n,m|k,l}},\ep)$, 
we need the asymptotic analysis of our distribution $P_{n,m,k,l}$
under the limit $n \to \infty$.  
The paper \cite{HHY2} studied the following two cases.
\begin{enumerate}
\item
 The limit $n \to \infty$ with $k,l$ and $m/n$ fixed. 
\item
 The limit $n \to \infty$ with $m/n$, $k/n$ and $l/n$ fixed. 
\end{enumerate}
We call them Type I and II limits, respectively.
The details are given in Subsection \ref{sss:intro:I} and 
Subsections \ref{sss:intro:II}. 
Hereafter we denote by $\Prb_n$ the probability for the distribution 
$P_{n,n \xi,n \kappa,n \alpha}$, and by $X$ a random variable distributed
subject to the distribution.
The contents of this section were already shown in the paper \cite{HHY2}.

\subsection{Asymptotic analysis of Type I limit}\label{sss:intro:I}
To address Type I limit, we employ the binomial distribution $B_{\xi,j}$ with $j$ trials
with successful probability $0 \le \xi \le 1$ 
whose pmf is given by $x \mapsto \binom{j}{x} \xi^x (1-\xi)^{j-x}$.
Then, we introduce the distribution
$ P_{R|\xi;k,l} := B_{\xi,k-l} * B_{1-\xi,l}$
where
$*$ denotes the convolution of probability distributions. 
Its pmf is denoted by $q(x \midd \xi;k,l)$
and written as 
\begin{align}\label{eq:I:q2}
\begin{split}
& q(x \midd \xi;k,l) \\
 =& \xi^{l-x} (1-\xi)^{k-l-x} \sum_{u=u_{*}}
 ^{u^* } \\
&\cdot \binom{k-l}{x-u} \binom{l}{u} \xi^{2(x-u)} (1-\xi)^{2u},
\end{split}
\end{align}
where 
$u_*:=\max (0 ,x-k+l)$
and $u^*:=\min( x , l)$.

\begin{thm}
\label{thm:intro:I}
In the limit $n \to \infty$ with $k,l$ and $\xi=m/n$ fixed, 
the distribution $P_{n,m,k,l}$ is approximated to 
the above-defined distribution
$ P_{R|\xi;k,l}$ 
up to $O(1/n)$.
In particular, the asymptotic expectation is given by
\begin{align}\label{eq:intro:Eq}
 \lim_{n \to \infty} \bbE[X] = (k-l)\xi+l(1-\xi).
\end{align}
\end{thm}

In fact, it is known as the law of small numbers 
 that the distribution $B_{c/n,n}$ converges to a Poisson distribution
 as $n$ goes to infinity with a fixed number $c$.
In Type I limit, $k$ and $l$ are fixed, and correspond to 
the fixed number $c$ of the law of small numbers. 
That is, we can consider Type I limit as 
the law of small numbers in our setting.

\begin{rem}
Since the limit pmf $q(x \midd \xi,k,l)$ is a convolution of two binomial distributions.
it is a polynomial of $\xi$,
and this property was used in the paper \cite{HY}.
\end{rem}

\subsection{Asymptotic analysis of type II limit}\label{sss:intro:II}
We discuss Type II limit.
For this aim, we 
employ three kinds of parametrizations of the fixed ratios
as Table \ref{tab:intro:IIprm}.
For example, the condition $(n,m,k,l) \in \clN$ in \eqref{eq:intro:clN} is equivalent to 
\begin{align}\label{eq:intro:II:pasmp}
 \alpha, \beta, \gamma,\delta \ge 0. 
\end{align}
However, in the following, we employ the set of three parameters 
$\beta,\delta,\xi \ge 0$, which are also free parameters
under the conditions $ \xi \ge \beta $ and $1-\xi \ge \delta$.
Also, the following quantities are fundamental for the discussion:
\begin{align}\label{eq:intro:msD}
\begin{split}
    \mu &:= \frac{1-\sqrt{D}}{2}, ~
 \sigma := \sqrt{\frac{(1-\beta-\delta)\beta\delta}{D}}, \\
      D &:= 4 \beta \delta+ (2\xi-1)^2. 
\end{split}
\end{align}
These definitions of these parameters are summarized 
in Table \ref{tab:intro:IIprm}.

\begin{thm}
\label{thm:intro:E}
Consider the limit $n \to \infty$ with fixed ratios $\alpha=\frac{l}{n}$, 
$\beta =\frac{m-l}{n}$, $\gamma=\frac{k-l}{n}$ and $\delta=\frac{n-m-k+l}{n}$.
Then, for any $\ve \in \bbR_{>0}$, we have 
\begin{align}
 \lim_{n \to \infty} \Prb_n\Bigl[\abs{\tfrac{X}{n}-\mu}>\ve\Bigr]=0.
\end{align}
In particular, the expectation $\bbE[X]$ behaves as 
\begin{align}
 \bbE[X] = n \mu + o(n).
\end{align}
\end{thm}

Considering the definition of the variance $\sigma^2$
by using $\beta,\delta,\xi \ge 0$, we find the following fact.
When the relation $\alpha=\gamma=0$ nor $\beta \delta=0$
does not hold, i.e., 
\begin{align}\label{eq:intro:II:asmp}
 \alpha+\gamma=1-\beta-\delta, \beta, \delta > 0,
\end{align}
the variance $\sigma^2$ is well defined by \eqref{eq:intro:msD}
and is strictly positive.
Then, we have the following theorem.

\begin{thm}
\label{thm:intro:CLT}
Assume the condition \eqref{eq:intro:II:asmp}. 
Then, for any real $t<u$, the relation
\begin{align}
 \lim_{n \to \infty} 
 \Prb_n\Bigl[t \le \frac{X-n \mu}{\sqrt{n} \sigma} \le u\Bigr]
 = \frac{1}{\sqrt{2\pi}} \int_t^u e^{-s^2/2} \, d s
\end{align}
holds. In particular, the pmf $p(x):=p(x \midd n,m,k,l)$ behaves as 
\begin{align}\label{eq:intro:II}
 \begin{split}
   p(x)& = \Psi(x) \cdot \bigl(1+o(1)\bigr) \quad (n \to \infty), \\
 \Psi(x) &:= \frac{1}{\sqrt{2\pi n}\sigma}
 \exp\biggl(-\frac{1}{2}\Bigl(\frac{x-n \mu}{\sqrt{n} \sigma}\Bigr)^2\biggr).
\end{split}
\end{align}
In addition, we have
\begin{align}
& \bbE[X] = n \mu + \phi + o(1), \quad 
 \phi := \frac{\sigma^2-\mu}{1-2\mu}, \\
& \lim_{n \to \infty} \bbV \Bigl[\frac{X}{\sqrt{n}}\Bigr] = 
 \lim_{n \to \infty} \bbE \Bigl[\Big(\frac{X-n\mu}{\sqrt{n}}\Big)^2\Bigr] 
 = \sigma^2 \\ 
 & \lim_{n \to \infty} \bbE \Bigl[\Big(\frac{X-n\mu}{\sqrt{n}}\Big)^j\Bigr] 
<\infty \label{NVT}
\end{align}
for $j=3,4, \ldots.$
\end{thm}

In the special case $\alpha \gamma=0$,
we can further discuss the tight exponential evaluation of the tail probability as follows.
\begin{prp}\label{prp:5.6.2}
In Type II limit with $\alpha \beta \delta >0$, $\gamma=0$ 
and $0< \xi = \alpha+\beta \le 1/2$, 
the probability $\Prb_n \bigl[| \tfrac{X}{n} - \mu |\ge \epsilon \bigr]$
goes to zero exponentially for any $\epsilon >0$.
\end{prp}

\section{Asymptotics analysis on $ S\bigl(\av_\pi[\Xi_{n,m|k,l}]\bigr) $}\label{S2-3}
As an interlude, let us derive the asymptotic formula \eqref{eq:intro:S-asymp}
of $S(\av_{\pi}[\Xi_{n,m|k,l}])$ for the state $\Xi_{n,m|k,l}$ 
with respect to the distribution $P_{n,m,k,l}$ in Type I limit
$n \to \infty$ with fixed $\xi=m/n$, $k$ and $l$.
We denote the pmf by $p(x):=p(x \midd n,m,k,l)$,
and the limit pmf by $q(x):=q(x \midd \xi;k,l)$ for simplicity.
Thus, we have $p(x) = q(x) + O(1/n)$ by Theorem \ref{thm:intro:I}.
Now, recall the expression \eqref{eq:intro:vNav} of $S(\av_{\pi}[\Xi_{n,m|k,l}])$,
which includes the dimension formula 
$\dim \clV_{(n-x,x)}=\binom{n}{x}\frac{n-2x+1}{n-x+1}$ in \eqref{eq:hook}.
Let us cite from \cite[Chap.\ 5, (5.12)]{Sp} 
the asymptotic formula of binomial coefficient $\binom{n}{x}$ for fixed $x$:
\begin{align}\label{eq:intro:binom}
\begin{split}
& \log \tbinom{n}{x} \\
=& 
 x \log n + x - (x+\tfrac{1}{2})\log x - \log \sqrt{2 \pi} + o(1). 
\end{split}
\end{align}
Applying it to the dimension formula, we show \eqref{eq:intro:S-asymp} as
\begin{align}\label{eq:S-asymp}
\begin{split}
& S(\av_{\pi}[\Xi_{n,m|k,l}]) \\
=&\sum_{x=0}^k p(x) \bigl(\log \dim \clV_{(n-x,x)} - \log p(x)\bigr) \\
=&\sum_{x=0}^k q(x) \bigl(x \log n + x - (x+\tfrac{1}{2})\log x \\
&- \log \sqrt{2\pi} - \log q(x) \bigr) + o(1) \\
=&u \log n + u - \log \sqrt{2\pi} \\
&+ 
 \sum_{x=0}^k q(x) \bigl((x+\tfrac{1}{2}) \log x-\log q(x) \bigr) + o(1),
\end{split}
\end{align}
where we used the asymptotic expectation $u:=(k-l)\xi+l(1-\xi)$
in \eqref{eq:intro:Eq}.
The last line in \eqref{eq:S-asymp} shows the conclusion \eqref{eq:intro:S-asymp}.
Therefore, 
with the comparison with \eqref{NBT} and \eqref{NHI1},
Dicke state $  \ket{\Xi_{N+M,M}}$ significantly improves
the amount of asymmetry over the classical symmetric state
$\rho_{mix,N+M,M}$.

\begin{figure}[htbp]
\centering
\includegraphics[width=.95\linewidth]{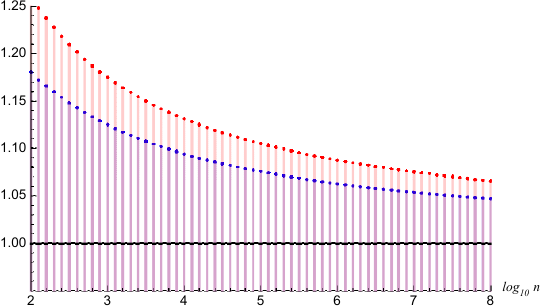}
\caption{Plots of $S\bigl(\av_\pi[\Xi_{n,m|k,l}]\bigr)/\log n$ (blue dots) 
 and the approximation $a(n,\xi,k,l)/\log n$ (red dots) 
 for $(\xi=m/n,k,l)=(0.5,2,1)$. 
This case implies $u:=(k-l)\xi+l(1-\xi)=1$. 
Here, $a(n,\xi,k,l)$ is defined as 
$u \log n + u - \log \sqrt{2\pi} 
+  \sum_{x=0}^k q(x) \bigl((x+\tfrac{1}{2}) \log x-\log q(x) \bigr)$ by using \eqref{eq:S-asymp}.
 The horizontal axis is $\log_{10} n$, and 
 the black line shows the limit value $u=1$.}
\label{fig:intro:vN}
\end{figure}

Continuing the interlude, we consider Type II limit of $S(\av_{\pi}[\Xi_{n,m|k,l}])$.
When $\beta=0$, i.e., $m-l=0$,  
we have $\av_\pi[\Xi_{n,m|k,l}]=\rho_{mix,n, m}$.
Hence, we have
\begin{align}
S(\av_\pi[\Xi_{n,m|k,l}])=S(\rho_{mix,n, l})=\log \binom{m}{n},\label{NZT}
\end{align}
which was already discussed in Section \ref{S2-15}.
The same discussion can be applied to the case with $\delta=0$.
When $\alpha+\gamma=0$, 
we have $\ket{\Xi_{n,m|k,l}}=
\ket{\Xi_{n,m}}$ so that $S(\av_\pi[\Xi_{n,m|k,l}])=
S(\ket{\Xi_{n,m}})=0$.
In this case, $S(\av_\pi[\kb{1^l \, 0^{k-l}}\otimes \rho_{mix,N+M,M}] )
-S(\kb{1^l \, 0^{k-l}}\otimes \rho_{mix,N+M,M} )$ is also zero.
These cases have no difference from the case in Section \ref{S2-15}.
Therefore, we assume the assumption \eqref{eq:intro:II:asmp} in the following.

We have
\begin{align}
\begin{split}
& S(\av_{\pi}[\Xi_{n,m|k,l}]) \\
 =&\sum_{x=0}^k p(x) 
 \Big(n h(\mu)+h'(\mu) (x-n\mu)\\
& +O((\frac{x-n\mu}{\sqrt{n}})^2)\Big)
 +O(\log n) ,\label{NBY4}
\end{split}
\end{align}
which is shown in Appendix \ref{NVR2}.
Theorem \ref{thm:intro:CLT} implies that
\begin{align}
 S(\av_{\pi}[\Xi_{n,m|k,l}]) =n h(\mu)+O(\log n).\label{NMT5}
\end{align}
When the relation $\alpha\gamma=0$ holds additionally,
the relation \eqref{NMT5}
is replaced by a more precise form as follows
\begin{align}
 S(\av_{\pi}[\Xi_{n,m|k,l}]) = n h(\mu)+C_0+o(1).\label{NHT}
\end{align}
This relation will be shown in Proposition \ref{prp:ProPH} in \ref{ss:II:sp},
and the the precise form of the constant $C_0$ is given in Proposition \ref{prp:ProPH}.

Due to \eqref{NBT} and \eqref{NHI2},
the comparison between the Dicke state and the decohered state is summarized to 
the comparison between their first-order terms
$h(\mu)$ and $h(\xi)-(\beta+\delta) h(\frac{\beta}{\beta+\delta})$.
In fact, the limit $n\to \infty$ in \eqref{NMT} with the division by $n$
implies
\begin{align}
h(\mu) \ge h(\xi)-(\beta+\delta) h(\frac{\beta}{\beta+\delta}).
\label{ZMY}
\end{align}
Here, we employ the relations \eqref{NHI2} and \eqref{NMT5}.
The difference between LHS and RHS shows the asymptotic difference 
of the amounts of the asymmetry activation
caused by the decoherence.
We consider the case when $\tau=0$, in which 
two free parameters describe other parameters as
\begin{align}
\begin{split}
\alpha&=\xi \kappa,~
\beta=\xi(1-\kappa),\\
\gamma&= (1-\xi)\kappa,~
\delta=(1-\xi)(1-\kappa).\label{NMA}
\end{split}
\end{align}
Then, since $\frac{\beta}{\beta+\delta}=\xi$ and 
$D 
=4 (1-\xi)\xi(1-\kappa)^2+ (2\xi-1)^2$, we have 
\begin{align}
\begin{split}
& h(\mu)= h(\frac{1-\sqrt{\frac{(1-\kappa)(1-\xi)\xi(1-\kappa)^2}{4 (1-\xi)\xi(1-\kappa)^2+ (2\xi-1)^2}}}{2}), 
\\
&h(\xi)-(\beta+\delta) h(\frac{\beta}{\beta+\delta})=
\kappa h(\xi).
\end{split}
\end{align}
In this case, 
$h(\mu)$ and $h(\xi)-(\beta+\delta) h(\frac{\beta}{\beta+\delta})$ are plotted as
Fig. \ref{fig:asymmetry}.

\begin{figure}[htbp]
\centering
\includegraphics[width=.95\linewidth]{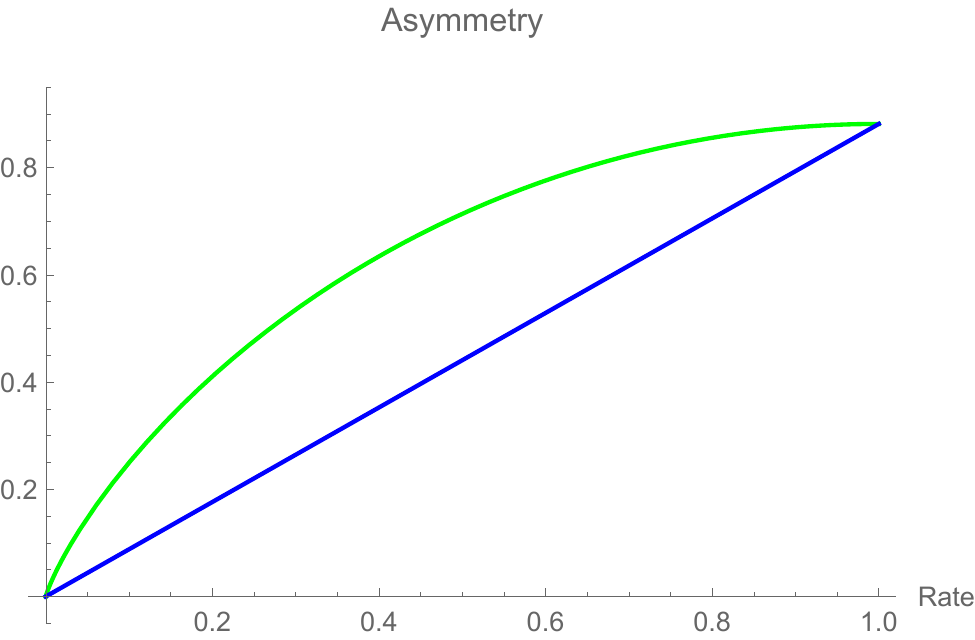}
\caption{Comparison between 
$h(\mu)$ and $h(\xi)-(\beta+\delta) h(\frac{\beta}{\beta+\delta})$.
We consider the case with \eqref{NMA}. 
The base of the logarithm is $2$. We set $\xi$ to be 0.3.
The horizontal line shows $\kappa$, which runs from $0$ to $1$.
The green line shows $h(\mu)$, and the blue line shows 
$h(\xi)-(\beta+\delta) h(\frac{\beta}{\beta+\delta})$.
}
\label{fig:asymmetry}
\end{figure}

\section{Asymptotics analysis on $ \log M({\Xi_{n,m|k,l}},\ep)$}\label{S2-4}
Next, to address the asymptotic behavior of $ \log M({\Xi_{n,m|k,l}},\ep)$,
we apply the discussion in Section \ref{S15} to our setting.
Since the state $\av_\pi[\Xi_{n,m|k,l}] $ is characterized as \eqref{ALO},
$H_s^{\ep}(\av_\pi[\Xi_{n,m|k,l}]) $ equals $\ol{H}_s^{\ve}(\av_\pi[\Xi_{n,m|k,l}])$
which is defined in \eqref{eq:intro:XAL}. Hence, 
\begin{align}
\begin{split}
 &H_s^{\ve}(\av_\pi[\Xi_{n,m|k,l}]) \\
 = &
 \max \bigl\{ \lambda \mid P_{n,m,k,l}[\{x \mid
  \log \dim \clV_{(n-x,x)}\\
&\hspace{10ex}  -\log p(x \midd n,m,k,l) \le \lambda\}] \le \ep \bigr\}.
\label{ZYS}
\end{split}
\end{align}

To discuss Type I limit, we denote the cdf of 
$B_{\xi,k-l}*B_{1-\xi,l} $ by $F[B_{\xi,k-l}*B_{1-\xi,l}]$. 
By recalling the dimension formula \eqref{eq:hook} of $\dim \clV_{(n-x,x)}$, 
$\log \dim \clV_{(n-x,x)}$ is approximated to $x \log n+o(\log n)$ 
with a fixed value $x$ by \eqref{eq:intro:binom}.
When $n$ is sufficiently large,
$\log \dim \clV_{(n-x,x)}-\log p(x \midd n,m,k,l) $
is monotonically increasing for $x$.
We choose $x_n(\epsilon)$ as $F[P_{n,m,k,l} ](x_n(\epsilon))=\epsilon$.
Hence, the combination of \eqref{ZYS} and Theorem \ref{thm:intro:I} implies
\begin{align}\label{AL1-BC}
\begin{split}
& H_s^{\ve}(\av_\pi[\Xi_{n,m|k,l}])  \\
 =& \log \dim \clV_{(n-x_n(\ep),x_n(\ep))} -\log p(x_n(\ep) \midd n,m,k,l)
\\
 =& x_n(\ep) \log n+o(\log n) -\log (q(x_n(\ep) \midd \xi;k,l)\\
 &+ O(1/n))
\\
 =& F[B_{\xi,k-l}*B_{1-\xi,l} ]^{-1}(\ep) \log n + o(\log n).
\end{split}
\end{align}
The combination of \eqref{EE9}, \eqref{EE10}, and 
\eqref{AL1-BC} yields
\begin{align}\label{AL1-B}
 \begin{split}
&\log M(\Xi_{n,m|k,l},\ep) \\
=& F[B_{\xi,k-l}*B_{1-\xi,l} ]^{-1}(\ep) \log n + o(\log n).
\end{split}
\end{align}

In Type II limit, 
when the assumption \eqref{eq:intro:II:asmp} does not hold,
the analysis on $ \log M(\Xi_{n,m|k,l},\ep) $ is the same as 
the case in Section \ref{S2-15}.
Therefore, we assume the assumption \eqref{eq:intro:II:asmp} in the following.
Then, as shown in Appendix \ref{NVR2}, we have
\begin{align}\label{AL2-BC}
\begin{split}
& H_s^{\ve}(\av_\pi[\Xi_{n,m|k,l}])  \\
=&n h(\mu)+h'(\mu) (x_n(\ep)-n\mu)\\
& +O((\frac{x_n(\ep)-n\mu}{\sqrt{n}})^2)+O(\log n).
\end{split}
\end{align}
Theorem \ref{thm:intro:CLT} implies that
\begin{align}
& H_s^{\ve}(\av_\pi[\Xi_{n,m|k,l}])  \nonumber \\
 =& h(\mu+ \frac{\Phi^{-1}(\ep)}{\sigma\sqrt{n}}+o(\frac{1}{\sqrt{n}}) ) n+O(\log n)
\nonumber\\
=& n h(\mu)  +\sqrt{n} h'(\mu)\frac{\Phi^{-1}(\ep)}{\sigma} 
 +o(\sqrt{n}). 
 \label{NMT5N}
\end{align}

The combination of \eqref{EE9}, \eqref{EE10}, 
\eqref{NMT5N}
yields
\begin{align}\label{AL2-B}
\begin{split}
& \log M(\Xi_{n,m|k,l},\ep) \\
 =&
 n h(\mu) +\sqrt{n} h'(\mu)\frac{\Phi^{-1}(\ep)}{\sigma}+ o(\sqrt{n}).
\end{split}
\end{align}
These relations \eqref{AL1-B} and \eqref{AL2-B} show 
the asymptotic behavior of the degree of the asymmetry.
The asymptotic expansion \eqref{NMT5N} has Gaussian distribution in the second order term, which is similar to 
the asymptotic distillation of the coherence \cite{HFW}.

\section{Another example of asymmetry activation}\label{S6}
To consider another simple example of asymmetry activation, 
we choose
$(\mathbb{C}^{n})^{\otimes n}$, $\SU(n)$, and 
the $n$-tensor product representation of $\SU(n)$
as ${\cal H}_1$, $G_1$, and $f_1$, respectively.
Then, the $n$-tensor anti-symmetric subspace ${\cal A}_{n,n}$
of $\SU(n)$ is a one-dimensional space.
We choose the unique state on 
the $n$-tensor anti-symmetric subspace of $\SU(n)$
as the invariant state $\rho_1$.
Given $d \ge n$,
we choose 
$(\mathbb{C}^{d})^{\otimes n}$, $\SU(d)$, and 
the $2$-tensor product representation of $\SU(d)$
as ${\cal H}_2$, $G_2$, and $f_2$, respectively.

We choose $ \SU(nd)$ and 
the $n$-tensor product representation of $\SU(nd)$
as $G_3$ and $f_3$, respectively. 
We choose the state $\rho_2$ as a state on 
the anti-symmetric subspace ${\cal A}_{d,n}$ of 
$(\mathbb{C}^{d})^{\otimes n}$.
Since the dimension of ${\cal A}_{d,n}$ is ${d \choose n}$,
$D(\rho_2\|\av_{f_2}[\rho_2])=\log {d \choose n}- S(\rho_2)$.
In this case, the vector of the support of $\rho_1\otimes \rho_2$ is 
invariant for the exchange of the first and second system.
Hence, the vector of the support belongs to 
the symmetric subspace ${\cal S}_{nd,n}$ of $(\mathbb{C}^{nd})^{\otimes n}$, whose dimension is
${nd+n-1 \choose n}$.
Thus, we have
$D(\rho_1\otimes \rho_2\|\av_{f_3}[\rho_1\otimes \rho_2])
=\log {nd+n-1 \choose n}- S(\rho_2)$.
Therefore, the amount of the asymmetry activation is 
\begin{align}
&\log {nd+n-1 \choose n}-\log {d \choose n}\nonumber \\
=&\sum_{j=0}^{n-1}\log \frac{nd+n-1-j}{d-j}\nonumber \\
=&\sum_{j=0}^{n-1}\log (n+\frac{(n-1)(j+1)}{d-j}).
\end{align}
Given a fixed integer $n$,
the amount of the asymmetry activation
takes the maximum $\log {n^2+n-1 \choose n}$ 
when $d=n$.
In this case, this value is approximated to be
$n\log n +O(n)$.

\section{Conclusion and discussion}\label{S5}
We have shown that the coherence in the Dicke state enhances
the asymmetry for the permutation when a qubits state $\ket{1^l \, 0^{k-l}}$
is attached.
To clarify the merit of the coherence, we have discussed two types of limits, Type I and II limits.
In these limits,
the case with the Dicke state has 
a larger amount of degree than 
the case with the decohered state.
In particular, under Type I limit,
the amount of the degree in the former has a strictly larger order than 
that of the latter.
Under Type II limit,
both cases have different leading terms.
Their difference characterizes the effect of the existence of
the coherence.
This fact shows the importance of the coherence even in the 
asymmetry and the asymmetry activation for permutation.

Here, we emphasize the generality of the concept of
asymmetry activation.
Asymmetry activation may happen in various situations.
To explain this aspect, we have shown another example in Section \ref{S6}.
It is another interesting future study to discuss various examples of asymmetry activation.

Using this discussion, we have derived the inequality \eqref{ZMY}.
Since the inequality \eqref{ZMY} can be considered without this type of limit,
it is possible to directly derive this inequality without considering this type of limit.
This kind of derivation is an interesting future study.
Further, our analysis has revealed the physical importance of the value
$\mu$, whose importance has not been recognized until this study.

For our asymptotic derivations,
we have employed various asymptotic formulas for 
the distribution $P_{n,m,k,l}$.
Since the distribution $P_{n,m,k,l}$ has a highly complicated form,
this fact shows the usefulness of 
these asymptotic formulas.
Since the distribution $P_{n,m,k,l}$ is based on the 
the Schur-Weyl duality, which is a key structure in quantum information,
we can expect that 
these formulas can be used in other topics in quantum information.

\appendix
\section{Type I and Type II limits and
Limit theorems for binomial distributions}\label{S-binomial}
First, we summarize existing results for binomial distributions.
Given a real number $p \in (0,1)$, we define the binomial distribution 
as
\begin{align}
P_{n,p}(k):= {n \choose k}p^{k}(1-p)^{n-k}.
\end{align}
The central limit theorem addresses the case when 
$p$ is fixed and $n$ increases.
It states
\begin{align}
&\lim_{n\to \infty }\sum_{k= np+ x_0 \sqrt{n}}^{np+ x_1 \sqrt{n}}
P_{n,p}(k) \notag \\
=& 
\int_{x_0}^{x_1} \frac{1}{\sqrt{2\pi p(1-p)}}e^{-\frac{x^2}{2p(1-p)}}.
\end{align}
However, 
the above convergence is not uniform when $p$ is close to $0$ or $1$.
To cover the asymptotic behavior in this case,
we consider the case when $p$ is chosen to $\lambda/n$
and $n$ increases.
Then, we have the following relation
\begin{align}
\lim_{n\to \infty }
P_{n,\lambda/n}(k)= 
e^{-\lambda}\frac{\lambda^k}{k!},
\end{align}
which is called the law of small numbers.

In our model, $m=M+l$ expresses the total number of $1$ and
$n=N+m+k-l$ expresses the total number of $0$.
Hence, the ratio $\frac{m}{n+m}$ corresponds to the parameter $p$ 
in the binomial distribution.
Hence, Type I and Type II limits corresponds to 
the law of small number and
the central limit theorem, respectively.
Therefore,
although Type II limit works with
$\xi=\frac{m}{n}>0$, 
when $\xi=\frac{m}{n}$ is close to zero,
Type II limit does not work so that 
Type I is needed for this case.

\section{Single-shot bounds for cq-channels}\label{s:SS}

The aim of this section is to derive the inequalities 
\eqref{EE9} and \eqref{EE10}, which will be shown in Theorem \ref{TH1}.
We begin with a general setting of classical-quantum channel explained in Subsection \ref{S2-4}.

\subsection{Single-shot bound for general cq-channel}\label{SS1}
We consider a pure state channel $\clX \ni x \mapsto W_x$ 
where $W_x$ is a state on a Hilbert space $\mathcal{H}$.
Let $M(\rho,\ep)$ be the number of distinguished states 
among $\{W_x\}_{x \in \clX}$.
We consider a joint system $\clH_\clX \otimes \clH$, 
where $\clH_\clX$ is a Hilbert space 
spanned by $\{ \ket{x} \}_{x \in \clX}$.
Dependently of a distribution $P$ on $\clX$ and a state $\rho$ on $\clH$,
we define two states $R[P]$ and $S[P,\rho]$ 
on the joint system $\clH_\clX \otimes \clH$ by
\begin{align}
\begin{split}
 R[P] &:= \sum_{x} P(x) \kb{x} \otimes W_x, \\
 S[P,\rho] &:= \Bigl( \sum_{x} P(x)\kb{x} \Bigr) \otimes \rho,
\end{split}
\end{align}
and a state $W_P$ on $\mathcal{H}$ by 
\begin{align}
 W_P:=\sum_{x} P(x) W_x.
\end{align}
Recall the information spectrum relative entropy $D_s^\delta$ 
in \eqref{eq:spectrum_rel_ent}:
\begin{equation}\label{EE1}
 D_s^\delta(\rho\|\rho'):= \max \bigl\{
 \lambda \mid \tr \rho \{ \rho \le e^\lambda \rho'\} \le \delta\bigr\},
\end{equation}
The hypothesis testing relative entropy $D_H^{\ep}$ is defined by~\cite{wang10} 
\begin{align}\label{EE2}
\begin{split}
 &D_H^{\ep}(\rho\|\rho') \\
 :=& 
 - \log \min \bigl\{ \tr {Q \rho'} \mid 0 \le Q \le I,\tr{Q\rho} \ge 1-\ep \bigr\}.
\end{split}
\end{align}
Nagaoka \cite{Nagaoka1} essentially showed the following inequality for any channel.
\begin{lem}\label{LL1}
For any state $\rho'$, we have
\begin{align}\label{EE3}
 M(\rho,\ep) \le \max_{P}D_H^{\ep}(R[P]\|S[P,\rho']).
\end{align}
\end{lem}

By \cite[Corollary 1]{beigi2014quantum}, 
we have the following inequality for any mixed state $\bar\rho$:
\begin{lem}\label{LL2}
The inequality
\begin{align}\label{EE4}
 M(\rho,\ep) \ge \max_{P}
 D_s^{\delta}(R[P]\|S[P, W_P])- \log \frac{1-\ep}{\ep-\delta} 
\end{align}
holds.
\end{lem}
By \cite[(22)]{Nagaoka} and \cite[Lemma 12]{tomamichel12}, we also have:
\begin{lem}\label{LL3}
The inequality
\begin{align}\label{EE5}
     D_s^{\ep}(\rho\|\rho')
 \le D_H^{\ep}(\rho\|\rho')
 \le D_s^{\ep+\delta}(\rho\|\rho')-\log \delta
\end{align}
holds.
\end{lem}

\subsection{Single-shot bound for pure state channel}\label{SS2}
In the case when the states $W_x$ are pure states, we have the following theorem.
\begin{thm}\label{TH1}
When all states $W_x$ are pure states, we have the following relations.
\begin{align}
 M(\rho,\ep) & \ge \max_{P} H_s^{\ep- \delta_1-\delta_2}(W_P)
                 - \log \frac{1-\ep}{\delta_1\delta_2},  
 \label{EE6} \\
 M(\rho,\ep) & \le \max_{P} H_s^{\ep+\delta_1+2\delta_2}(W_P) 
                 + \log \frac{1}{\delta_1 \delta_2^2}.
 \label{EE7}
\end{align}
\end{thm}

To show Theorem \ref{TH1}, we prepare the following Lemmas \ref{L1} and \ref{L2},
using the fact that $W_x$ is a pure state.
Although these two lemmas were essentially proved in \cite[Appendix II]{Ha10-1},
it is difficult to extract them from the reference, 
and we present them here with proofs.

\begin{lem}\label{L1}
The relation
\begin{align}\label{EE8}
 D_H^{\ep}(R[P]\|S[P,W_P]) \ge H_s^{\ep} (W_P)
\end{align}
holds.
\end{lem}

\begin{proof}
We define the projection $W_x''$ as
\begin{align}
& W_x'' \nonumber \\
 :=&
 \begin{cases}
 \frac{1}{\tr W_x \{ I > e^{\lambda} W_P \}}
 \{ I > e^{\lambda} W_P \}W_x\{ I > e^{\lambda} W_P \}\\
    \hspace{10ex} \text{ when } \tr W_x \{ I \ge e^{\lambda} W_P \} \neq 0 \\
 0 \\
  \hspace{10ex} \text{ when } \tr W_x \{ I \ge e^{\lambda} W_P \}   =  0 .
 \end{cases}
\end{align}
Hence, we have 
\begin{align}
 \tr W_x  \{ I > e^{\lambda} W_P \} = \tr W_x  W_x'',
\end{align}
and thus
\begin{align}
& \tr W_P \{ I \ge e^{\lambda} W_P \} \nonumber\\
 =& \sum_{x}P(x) \tr W_x  \{ I \ge e^{\lambda} W_P \} 
\nonumber\\
 =& \sum_{x}P(x) \tr W_x  (I-W_x'').\label{NU1}
\end{align}
Since the support of $ W_x''$ is included in the support of 
$\{ I > e^{\lambda} W_P \} $ and $W_x''$ is the zero matrix or a rank-one projection,  
$W_x'' W_P W_x''$ is the zero matrix or a rank-one matrix
whose eigenvalue is not greater than $e^{-\lambda}$.
Hence, 
\begin{align}\label{NU2}
 \tr W_P W_x'' \le e^{-\lambda}.
\end{align}
We define the projection $T$ as
\begin{align}
 T:= \sum_{x} \kb{x} \otimes W_x''.
\end{align}
Then, using \eqref{NU1} and \eqref{NU2}, we have
\begin{align*}
 \tr S[P,W_P] T& \le e^{-\lambda}, \\
 \tr R[P] (I-T) &= \tr W_P \{ I \ge e^{\lambda} W_P \}. 
\end{align*}
\end{proof}

\begin{lem}\label{L2}
We have
\begin{align}
 &D_s^{\ep}\Big(R[P]\Big\| S\Big[P,\sum_x P(x) W_x\Big]\Big) 
\nonumber\\
 \le & 
 H_s^{\ep+2 \delta}\Bigl(\sum_x P(x) W_x\Bigr) +2 \log \frac{1}{\delta}.
\label{NMP}
\end{align}
\end{lem}

\begin{proof}
In this proof, we set $W_P:= \sum_x P(x) W_x $.
We choose $\lambda':=2 \log \frac{1}{\delta}$ and 
assume that
\begin{align}
 \lambda= H_s^{\ep-\delta}\Big(  \sum_x P(x) W_x\Big).
\end{align}

We define the non-negative matrix $W_x'$ as
\begin{align}
 W_x':= \{ W_x > e^{\lambda+\lambda'} W_P\} 
\end{align}
Since $W_x$ is a rank-one projection,
$\{ W_x > e^{\lambda+\lambda'} W_P\}$ is a rank-one projection or $0$, 
which implies $\tr  W_x'  W_x \le 1$.
Since $\tr  W_x' ( W_x - e^{\lambda+\lambda'} W_P) \ge 0$, we have
\begin{align}
 \tr W_x' (e^{\lambda+\lambda'} W_P)\le \tr  W_x'  W_x \le 1\label{AO}.
\end{align}
Then, we have 
\begin{align}\label{AO1}
 \tr W_x' (e^{\lambda} W_P) \le e^{-\lambda'}.
\end{align}
Also we have
\begin{align}\label{AO2}
 \{ I \le e^{\lambda} W_P \} \le e^{\lambda} W_P.
\end{align}
The combination of \eqref{AO2} and \eqref{AO} implies
\begin{align}\label{AO3}
 \tr W_x' \{ I\le e^{\lambda} W_P \} \le 
 \tr W_x' (e^{\lambda} W_P) \le e^{-\lambda'}.
\end{align}

Now, we have
\begin{align}
\begin{split}
& \bigl\| W_x' - \{ I > e^{\lambda} W_P\}   W_x' \{I > e^{\lambda} W_P \} \bigr\|_1 \\
=&\Big\| \{ I > e^{\lambda} W_P\}W_x' \{I \le e^{\lambda} W_P \} 
\\
&\hspace{10ex}+
         \{ I \le e^{\lambda} W_P\}W_x' \Big\|_1 \\
\le & \left\| \{ I > e^{\lambda} W_P\}W_x' \{I \le e^{\lambda} W_P \} \right\|_1 \\
&+ 
     \left\| \{ I \le e^{\lambda} W_P\}W_x' \right\|_1 \\
\stackrel{(a)}{\le} &2 \sqrt{\tr \{I \le e^{\lambda} W_P \}W_x'}
 \stackrel{(b)}{\le} 2 e^{-\lambda'/2},
\end{split}
\end{align}
where $(a)$ follows from Winter's gentle measurement lemma 
\cite{Winter}\footnote{Its detail is also available from \cite[Exercise 6.8 and Solution to Exercise 6.8]{Springer}.}, and $(b)$ follows from \eqref{AO3}.
Thus, 
\begin{align}
&     \tr W_x  \{ W_x > e^{\lambda+\lambda'} W_P\}
=   \tr W_x W_x' \nonumber \\
 \le &\tr W_x \{ I > e^{\lambda} W_P \}W_x' \{I \ge e^{\lambda} W_P \}
     + 2 e^{-\lambda'/2} \nonumber \\
\le &\tr W_x \{ I > e^{\lambda} W_P \}+ 2 e^{-\lambda'/2}
 \nonumber \\
 =&   \tr  W_x \{ I > e^{\lambda} W_P \}+ 2 \delta.
\end{align}
Then, we have 
\begin{align}
\begin{split}
& \tr R[P] \Big\{ R[P] \le e^{\lambda+\lambda'} S\Big[P,\sum_x P(x) W_x\Big]  \Big\} \\
&=   \sum_{x}P(x) \tr W_x  \{ W_x \le e^{\lambda+\lambda'} W_P\}\\
 &=1- \sum_{x}P(x) \tr W_x  \{ W_x   > e^{\lambda+\lambda'} W_P\} \\
&\ge 
  1- \Big(\sum_{x}P(x) \tr  W_x \{ I > e^{\lambda} W_P \}\Big)- 2 \delta \\
 &=   \Big(\sum_{x}P(x) \tr  W_x \{ I \le e^{\lambda} W_P \}\Big)- 2 \delta \\
&=\tr \Big( \bar{W} \{ I \le e^{\lambda} W_P \}\Big)- 2 \delta,
\end{split}
\end{align}
which implies \eqref{NMP}.
\end{proof}

\begin{proof}[{Proof of Theorem \ref{TH1}}]
We will show Theorem \ref{TH1} by 
combining Lemmas \ref{L1} and \ref{L2} to Lemmas in Appendix \ref{SS1}.
The inequality \eqref{EE6} is shown as follows:
\begin{align}
\begin{split}
&  M(\rho,\ep)\\
& \stackrel{(a)}{\ge} \max_{P} D_s^{\ep- \delta_1}(R[P]\|S[P,W_P])
  - \log \frac{1-\ep}{\delta_1} \\
& \stackrel{(b)}{\ge} \max_{P} D_H^{\ep- \delta_1-\delta_2}(R[P]\|S[P,W_P])
  - \log \frac{1-\ep}{\delta_1\delta_2} \\
& \stackrel{(c)}{\ge} \max_{P} H_s^{\ep- \delta_1-\delta_2}(W_P)
  - \log \frac{1-\ep}{\delta_1\delta_2},
\end{split}
\end{align}
where steps $(a)$, $(b)$ and $(c)$ follow from \eqref{EE4} of Lemma \ref{LL2}, 
{}from the second equation of \eqref{EE5} of Lemma \ref{LL3}, 
and from \eqref{NMP} of Lemma \ref{L2}, respectively.
The inequality  \eqref{EE7} is shown as follows:
\begin{align}
\begin{split}
&  M(\rho,\ep) \\
& \stackrel{(a)}{\le} \max_{P} D_H^{\ep}(R[P]\|S[P,W_P])\\
&  \stackrel{(b)}{\le} \max_{P} D_s^{\ep+\delta_1}(R[P]\|S[P,W_P]) - \log \delta_1 \\
& \stackrel{(c)}{\le} \max_{P} H_s^{\ep+\delta_1+2\delta_2}(W_P) 
  + \log \frac{1}{\delta_1 \delta_2^2}.
\end{split}
\end{align}
In the step $(a)$, we applied \eqref{EE3} of Lemma \ref{LL1} to the case with $\rho'=W_P$.
The steps $(b)$ and $(c)$ follow from the second inequality in \eqref{EE5} 
of Lemma \ref{LL3}, and from \eqref{EE8} of Lemma \ref{L1}, respectively.
\end{proof}

\subsection{Single-shot bound for pure state covariant channel}\label{SS3}

Now, we consider the group covariant case, i.e., 
the case when the set of states $\{W_x\}_{x \in \clX}$ 
is given as $\{ f(g) \rho f(g)^\dagger \}_{g \in G}$, where
$G$ is a finite group, $f$ is its unitary representation, and $\rho$ is a pure state. 
In this case, Theorem \ref{TH1} is simplified as follows.
The following theorem shows \eqref{EE9} and \eqref{EE10}.

\begin{thm}\label{TH2}
When $\{W_x\}_{x \in \clX}$ is given as the above defined set 
$\{ f(g) \rho f(g)^\dagger \}_{g \in G}$, we have the following relations.
\begin{align}
 M(\rho,\ep) &\ge H_s^{\ep-\delta_1-\delta_2}(\av_f[\rho]) 
                  - \log \frac{1}{\delta_1 \delta_2}, \\
 M(\rho,\ep) &\le H_s^{\ep+\delta_1+2\delta_2}(\av_f[\rho]) 
                  + \log \frac{1}{\delta_1 \delta_2^2}.
 \label{SS:EE10}
\end{align}
\end{thm}

\begin{proof}
Since 
$\tr \rho \{ \rho \le e^\lambda \av_f[\rho] \}= 
\tr f(g)\rho f(g)^\dagger\{ f(g)\rho f(g)^\dagger\le e^\lambda \av_f[\rho] \}$, we have
\begin{align}\label{EE11}
\begin{split}
& D_s^{\ep+\delta_1}(R[P]  \|S[P,\av_f[\rho]]) \\
=&
 D_s^{\ep+\delta_1}(R[P_U]\|S[P_U,\av_f[\rho]]),
\end{split}
\end{align}
where $P_U$ is the uniform distribution on $G$.
In the following step $(a)$, we apply  \eqref{EE3} of Lemma \ref{LL1} to 
the case with $\rho'=\av_f[\rho]$:
\begin{align}
\begin{split}
& M(\rho,\ep) \\
&\stackrel{(a)}{\le} \max_{P}D_H^{\ep}(R[P]\|S[P,\av_f[\rho]])\\
& \stackrel{(b)}{\le} \max_{P} D_s^{\ep+\delta_1}(R[P]\|S[P,\av_f[\rho]]) 
                   - \log \delta_1 \\
&\stackrel{(c)}{=}   D_s^{\ep+\delta_1}(R[P_U]\|S[P_U,\av_f[\rho]]) - \log \delta_1 \\
& \stackrel{(d)}{\le} \max_{P} H_s^{\ep+\delta_1+2\delta_2}(\av_f[\rho]) 
                   + \log \frac{1}{\delta_1 \delta_2^2},
\end{split}
\end{align}
where steps $(b)$, $(c)$ and $(d)$ follow from 
the second inequality of \eqref{EE5} of Lemma \ref{LL3},
{}from \eqref{EE11}, and from \eqref{EE8} of Lemma \ref{L1}, respectively.

By \cite[Theorem 2]{Korzekwa}, we have the following evaluation:
\begin{equation}\label{SS:direct}
 \log M(\rho,\ep) \ge
 D_s^{\ep-\delta}(  \rho\| \av_f[\rho] )) - \log \frac{1}{\delta} .
\end{equation}
Applying to \eqref{SS:direct} the same discussion as the proof of \eqref{EE6},
we obtain \eqref{SS:EE10}.
\end{proof}

\section{Derivations of \eqref{NHI1} and \eqref{NHI2}}\label{NVR}
We derive the formulas \eqref{NHI1} and \eqref{NHI2}.
For this aim, we recall the following formula when
$x/n$ is fixed and $n \to \infty$;
\begin{align}
& \log \tbinom{n}{x} 
 =n h(\tfrac{x}{n})+\tfrac{1}{2}\log \tfrac{n}{2\pi x(n-x)}+o(1)\nonumber \\
  =&n h(\tfrac{x}{n})-\tfrac{1}{2}\log n
     -\tfrac{1}{2}\log \bigl(\tfrac{2\pi x}{n}(1-\tfrac{x}{n})\bigr)+o(1),\label{NBY}
\end{align}
Under the Type I limit, using \eqref{NBY}, we have
\begin{align}
&\log \binom{m-l}{n-k}\nonumber \\
=
&(n-k) h(\frac{n\xi -l}{n-k})-\frac{1}{2}\log (n-k)
\nonumber \\
&-\tfrac{1}{2}\log \bigl( 2\pi \frac{n\xi -l}{n-k}
(1-\frac{n\xi -l}{n-k})\bigr)+o(1) \nonumber \\
= &(n-k) h(\frac{\xi -l/n}{1-k/n})-\frac{1}{2}\log (n-k)
\nonumber \\
&-\tfrac{1}{2}\log \bigl( 2\pi \frac{\xi -l/n}{1-k/n}
(1-\frac{\xi -l/n}{1-k/n})\bigr)+o(1) \nonumber \\
= &(n-k) h((\xi -l/n)(1+k/n)+o(1/n))
\nonumber \\
&
-\frac{1}{2}\log (n-k) 
-\tfrac{1}{2}\log \bigl( 2\pi \xi(1-\xi)\bigr)+o(1)\nonumber  \\
= &(n-k) h( \xi +\frac{\xi k-l}{n}
+o(1/n))
-\frac{1}{2}\log (n-k) 
\nonumber \\
&-\tfrac{1}{2}\log \bigl( 2\pi \xi(1-\xi)\bigr)+o(1) \nonumber \\
= &(n-k) h( \xi)+
(1-\frac{k}{n}) h'( \xi)(\xi k-l)
\nonumber \\
&-\frac{1}{2}\log n (1-\frac{k}{n}) -\tfrac{1}{2}\log \bigl( 2\pi 
(\xi  (1-\xi)\bigr)+o(1) \nonumber \\
= &(n-k) h( \xi)
-\frac{1}{2}\log n
+\frac{1}{2} \cdot \frac{k}{n} 
\nonumber \\
&-\tfrac{1}{2}\log \bigl( 2\pi \xi  (1-\xi)\bigr)
+h'( \xi)(\xi k-l)
+o(1) .
\end{align}
Therefore, we obtain \eqref{NHI1} as follows.
\begin{align}
&\log \binom{m}{n}-\log \binom{m-l}{n-k}\nonumber \\
=&
nh(\xi)-\frac{1}{2}\log n
-\tfrac{1}{2}\log \bigl( 2\pi \xi(1-\xi)\bigr)+o(1) \nonumber \\
&-\Big( (n-k) h( \xi)
-\frac{1}{2}\log n
-\tfrac{1}{2}\log \bigl( 2\pi \xi  (1-\xi)\bigr)
\nonumber \\
&
+h'( \xi)(\xi k-l)
+o(1) \Big)\nonumber \\
=& k h( \xi)
- h'( \xi)(\xi k-l)
+o(1) .
\end{align}

In Type II limit, using \eqref{NBY}, we obtain \eqref{NHI2} as
\begin{align}
&\log \binom{m}{n}-\log \binom{m-l}{n-k} \nonumber\\
=&
nh(\xi)-\frac{1}{2}\log n
-\tfrac{1}{2}\log \bigl( 2\pi \xi(1-\xi)\bigr)+o(1) \nonumber\\
&- \Big(n(\beta+\delta) h(\frac{\beta}{\beta+\delta})-\frac{1}{2}\log n(\beta+\delta)
\nonumber \\
&
-\tfrac{1}{2}\log \bigl( 2\pi \frac{\beta}{\beta+\delta}(1-\frac{\beta}{\beta+\delta})\bigr)+o(1) \Big) \nonumber\\
=&
n(h(\xi)-(\beta+\delta) h(\frac{\beta}{\beta+\delta}))
+\frac{1}{2}\log (\beta+\delta)
\nonumber \\
&
-\tfrac{1}{2}\log \bigl( 2\pi \xi(1-\xi)\bigr)  \nonumber\\
&+\tfrac{1}{2}\log \bigl( 2\pi \frac{\beta}{\beta+\delta}(1-\frac{\beta}{\beta+\delta})\bigr)+o(1) .
\end{align}

\section{Derivations of \eqref{NBY4} and \eqref{AL2-BC}}\label{NVR2}
We derive the formulas \eqref{NBY4} and \eqref{AL2-BC}.
In Type II limit, using \eqref{NBY}, we obtain 
\eqref{NBY4} as
\begin{align}
& S(\av_{\pi}[\Xi_{n,m|k,l}]) 
\nonumber \\
=&\sum_{x=0}^k p(x) \bigl(\log \dim \clV_{(n-x,x)} - \log p(x)\bigr) \nonumber\\
\stackrel{(a)}{=}&\sum_{x=0}^k p(x) \log \dim \clV_{(n-x,x)} 
 +O(\log k) \nonumber\\
\stackrel{(b)}{=}&\sum_{x=0}^k p(x) \log \dim \binom{n}{x}\frac{n-2x+1}{n-x+1}
 +O(\log k) \nonumber \\
\stackrel{(c)}{=}&\sum_{x=0}^k p(x) nh(\frac{x}{n})
 +O(\log n) \nonumber\\
 =&\sum_{x=0}^k p(x) n h(\mu+ \frac{x-n\mu}{n})
 +O(\log n) \nonumber\\
 =&\sum_{x=0}^k p(x) 
 \Big(n h(\mu)+h'(\mu) (x-n\mu)\nonumber \\
&
 +O((\frac{x-n\mu}{\sqrt{n}})^2)\Big)
 +O(\log n) ,\label{NBY2}
\end{align}
where each step is derived as follows.
$(a)$ follows from the relation 
$-\sum_{x=0}^k p(x) \log p(x) \le \log k$, 
$(b)$ follows from the relation
$\dim \clV_{(n-x,x)}=\binom{n}{x}\frac{n-2x+1}{n-x+1}$ in \eqref{eq:hook}.
$(c)$ follows from the relation \eqref{NBY}.

Next, we show \eqref{AL2-BC}.
In Type II limit, 
to evaluate the value $H_s^{\ve}(\av_\pi[\Xi_{n,m|k,l}])$, 
we introduce the following quantities.
\begin{align}
 &H_s^{\ve,1}(\av_\pi[\Xi_{n,m|k,l}]) \nonumber \\
  := &
 \max \Bigl\{ \lambda \Big|
 \nonumber \\
&
P_{n,m,k,l}[\{x \mid
  \log \dim \clV_{(n-x,x)} \le \lambda\}] \le \ep \Bigr\}.
\label{ZYS1} \\
 & H_s^{\ve,2}(\av_\pi[\Xi_{n,m|k,l}]) \nonumber \\
:= &
 \max \Bigl\{ \lambda \mid \nonumber \\
&
P_{n,m,k,l}[\{x \Big|
  -\log p(x \midd n,m,k,l) \le \lambda\}] \le \ep \Bigr\}.
\label{ZYS2}
\end{align}
We have
\begin{align}
 & H_s^{\ve_1,1}(\av_\pi[\Xi_{n,m|k,l}]) \le
 H_s^{\ve_1}(\av_\pi[\Xi_{n,m|k,l}]) 
\nonumber \\
 \le &
 H_s^{\ve_1+\ve_2,1}(\av_\pi[\Xi_{n,m|k,l}]) +
 H_s^{\ve_2,2}(\av_\pi[\Xi_{n,m|k,l}]) .
\end{align}
Also, we have
\begin{align}
0&\le H_s^{\ve,2}(\av_\pi[\Xi_{n,m|k,l}]) 
\nonumber \\
&\le -\log \frac{\ve}{k}
=-\log \ep+\log k .
\end{align}
Thus,
\begin{align}
& H_s^{\ve_1,1}(\av_\pi[\Xi_{n,m|k,l}]) 
 \le 
 H_s^{\ve_1}(\av_\pi[\Xi_{n,m|k,l}]) \nonumber \\
\le &
 H_s^{\ve_1+\ve_2,1}(\av_\pi[\Xi_{n,m|k,l}]) -\log \ep_2+\log k .\label{BFAI}
\end{align}
Also, we denote by $\Phi$ and $\tilde{\Phi}$
the cdfs of the standard Gaussian distribution and 
the Rayleigh distribution, respectively.

By \cite[Lemma 4.7.1]{As}, 
we have under the limit $n \to \infty$ with fixed ratio $t=x/n$ that 
\begin{align*}
 \log \binom{n}{x} &= h(t) n + O(\log n), \\
 h(t)&:=-t \log t -(1-t) \log(1-t).
\end{align*}
Then, we can approximate $\log \dim \clV_{(n-x,x)}$ to $h(\frac{x}{n}) n+O(\log n)$ 
when $x$ behaves linearly with respect to $n$.
Hence,
\begin{align}\label{AL2-BC6}
& H_s^{\ve,1}(\av_\pi[\Xi_{n,m|k,l}])  
 = \log \dim \clV_{(n-x_n(\ep),x_n(\ep))}
\nonumber \\
=& h(\frac{x_n(\ep)}{n}) n+O(\log n)
\nonumber\\
=&n h(\mu)+h'(\mu) (x_n(\ep)-n\mu)\nonumber \\
&
 +O((\frac{x_n(\ep)-n\mu}{\sqrt{n}})^2)+O(\log n).
\end{align}
Combining \eqref{BFAI} and \eqref{AL2-BC6}, we have
\eqref{AL2-BC}, i.e., 
\begin{align}
& H_s^{\ve}(\av_\pi[\Xi_{n,m|k,l}])  \nonumber \\
=&n h(\mu)+h'(\mu) (x_n(\ep)-n\mu)\nonumber \\
&
 +O((\frac{x_n(\ep)-n\mu}{\sqrt{n}})^2)+O(\log n).
\end{align}

\section{von Neumann entropy $S(\av_f[\psi])$
in special case $\alpha \gamma=0$}\label{ss:II:sp}
The aim of this appendix is to show \eqref{NHT}
by assuming
\eqref{eq:intro:II:asmp} and $\alpha \gamma=0$.
In this appendix, we directly derive \eqref{NHT} without use of \eqref{NBY4}.
For this aim,
to address a more detailed analysis than Section \ref{S2-2},
we consider Type II limit in the case $\alpha \gamma=0$,
corresponding to the case $K L=0$. 
We use Proposition \ref{prp:ProPH}, which 
discusses the tight exponential evaluation of the tail probability.
Using it, we derive
the asymptotic behavior of the von Neumann entropy $S(\av_f[\psi])$.

Before taking the limit $n \to \infty$, 
the parameters concerned are $(n,m,k,l)$ satisfying $K:=k-l=0$ or $L:=l=0$.
Let us focus on the former case $k=l$, and also assume $m \le n-m$.

Then, we prepare the following proposition, which was shown in \cite{HHY1}.
\begin{prp}\label{prp:k=l}
Under the assumption $m \le n-m$ and $k=l$, we have 
\begin{align}
 p(x) = \frac{\binom{n}{x}}{\binom{n}{m}} \frac{\binom{l}{x}}{\binom{m}{x}}
        \frac{n-2x+1}{n-x+1} \binom{n-l-x}{m-l}
\label{eq:ac=0:p}
\end{align} 
for $x \le l \wedge m$, and $p(x)=0$ otherwise.
\end{prp}

Below we assume $l>0$, $M:=m-l>0$ and $N:=n-m-k+l>0$
because of \eqref{eq:intro:II:asmp}.
In this section, we consider the Type II limit 
with $\gamma=0$, i.e., the limit $n \to \infty$ fixed ratios 
\begin{align*}
& \gamma=\tfrac{k-l}{n}=0, \ \alpha=\tfrac{l}{n}>0, \
 \beta=\tfrac{M}{n}>0, 
\nonumber \\
&
  \delta=\tfrac{N}{n}>0, \
\xi=\alpha+\beta=\tfrac{m}{n} \le 1/2.
\end{align*} 
The quantities are then given by 
\begin{align}\label{eq:ac=0:munu}
\begin{split}
& D = 1-4 \alpha(1-\xi), \\
&\mu = \frac{1-\sqrt{1-4\alpha(1-\xi)}}{2}, \\
& \nu = \frac{1+\sqrt{1-4\alpha(1-\xi)}}{2}, \quad 
 \sigma = \sqrt{\frac{\alpha \beta \delta}{D}}.
\end{split}
\end{align}
Note that the assumption $0<\xi \le 1/2$ and $0<\alpha<1/2$ yields $D, \sigma >0$ and 
\begin{align}\label{eq:ac=0:mhn}
 \mu < \tfrac{1}{2} < \nu.
\end{align}
Also, by Theorem \ref{thm:intro:CLT}, the expectation and the variance asymptotically behave as
\begin{align}
& \bbE[X] = n \mu + \phi + o(1),  \label{eq:ac=0:NJ1} 
 \\
& \phi := \frac{\sigma^2-\mu}{1-2\mu} = 
 \frac{1}{\sqrt{D}} (\frac{\alpha \beta \delta}{D}-\frac{1-\sqrt{D}}{2}), 
 \nonumber \\
 &\lim_{n \to \infty} \bbV \Bigl[\frac{X}{\sqrt{n}}\Bigr] = 
 \lim_{n \to \infty} \bbE \Bigl[\Big(\frac{X-n\mu}{\sqrt{n}}\Big)^2\Bigr] = 
 \sigma^2 = \frac{\alpha \beta \delta}{D}. 
 \label{eq:ac=0:NJ2}
\end{align}

Then, we show the following proposition to address the asymptotic expansion of 
$S(\av_f[\psi])$ by using the formulas \eqref{eq:ac=0:NJ1} and \eqref{eq:ac=0:NJ2}.
This proposition is the precise statement of \eqref{NHT}.
\begin{prp}\label{prp:ProPH}
In Type II limit with the condition $\alpha \beta \delta >0$, $\gamma=0$ 
and $0< \xi \le 1/2$, we have
\begin{align}\label{eq:ac=0:NT1}
 S(\av_f[\psi]) = n C_1+ C_2+C_3\phi+C_4\sigma^2+o(1),
\end{align}
where
\begin{align*}
C_1 :=& h(\xi) + \xi h\Bigl(\frac{\mu}{\xi}\Bigr)
             -\alpha h\Bigl(\frac{\mu}{\alpha}\Bigr)\\
             &-(1-\alpha-\mu) h\Bigl(\frac{\beta}{1-\alpha-\mu}\Bigr), \\
C_2 :=& \frac{1}{2}\log \frac{\alpha \beta}{\xi^2(1-\xi)}
        \Bigl(1-\frac{\beta}{1-\alpha-\mu}\Bigr), \\
C_3 :=& h'\Bigl(\frac{\mu}{\xi}\Bigr)
       -h'\Bigl(\frac{\mu}{\alpha}\Bigr)+h\Bigl(\frac{\beta}{1-\alpha-\mu}\Bigr)\\
      & -\frac{\beta}{1-\alpha-\mu} h'\Bigl(\frac{\beta}{1-\alpha-\mu}\Bigr), \\
C_4 :=& \frac{1}{2\xi} h''\Bigl(\frac{\mu}{\xi}\Bigr)
       -\frac{1}{2\alpha} h''\Bigl(\frac{\mu}{\alpha}\Bigr)
       \\&-\frac{\beta^2}{2(1-\alpha-\mu)^3} h''\Bigl(\frac{\beta}{1-\alpha-\mu}\Bigr).
\end{align*} 
\end{prp}

In fact, the leading coefficient $C_1$ equals $h(\mu)$ as follows.
\begin{lem}
We have $C_1 = h(\mu)=-\mu \log \mu - (1-\mu) \log(1-\mu)$.
\end{lem}

\begin{proof}
A direct calculation (regarding $\mu$ as an indeterminate) yields
\begin{align}\label{eq:ac=0:C13}
\begin{split}
 &C_1 -h(\mu) \\
 = 
&(1-\mu) \log(1-\mu)(1-\xi-\mu)(\alpha-\mu) 
\\&
+ \mu \log \mu(\xi-\mu)(1-\alpha-\mu) \\
&-(1-\alpha) \log(1-\alpha-\mu)(\alpha-\mu)\\
& - \xi \log(1-\xi-\mu)(\alpha-\mu) \\
&-\alpha \log \alpha-(1-\xi)\log(1-\xi)\\
&+(\xi-\alpha)\log(\xi-\alpha).
\end{split}
\end{align}
Since $\mu$ is a solution of the cubic equation 
$(1-t)(1-\xi-t)(\alpha-t)= t(\xi-t)(1-\alpha-t)$ for $t$, 
the first line of \eqref{eq:ac=0:C13} can be computed as 
\begin{align}
\begin{split}
& (1-\mu) \log(1-\mu)(1-\xi-\mu)(\alpha-\mu)\\
&+\mu \log \mu(\xi-\mu)(1-\alpha-\mu) \\
=&\tfrac{1}{2} \log(1-\mu)(1-\xi-\mu)(\alpha-\mu) \\
&+ 
  \tfrac{1}{2} \log \mu(\xi-\mu)(1-\alpha-\mu).
\end{split}
\end{align}
Then we can proceed as 
\begin{align}
\begin{split}
 &C_1 -h(\mu) \\
 = 
&\tfrac{1}{2} \log \mu(1-\mu) + 
 (\tfrac{1}{2}-\xi) \log (\xi-\mu)(1-\xi-\mu) \\
 &+ 
 (\alpha-\tfrac{1}{2}) \log (\alpha-\mu)(1-\alpha-\mu) 
-\alpha \log \alpha \\
&-(1-\xi)\log(1-\xi)+(\xi-\alpha)\log(\xi-\alpha) .
\end{split}
\end{align}
Then, using $\mu-\mu^2=\alpha(1-\xi)$, $(\alpha-\mu)(1-\alpha-\mu)=\alpha(\xi-\alpha)$
and $(\xi-\mu)(1-\xi-\mu)=(1-\xi)(\xi-\alpha)$, we can check $C_1-h(\mu)=0$.
\end{proof}

To show Proposition \ref{prp:ProPH}, we compute 
\begin{align*}
&\log \dim \clV_{(n-x,x)}-\log p(x)\\
=&\log  \binom{n}{m} \binom{m}{x} \binom{l}{x}^{-1} \binom{n-l-x}{m-l}^{-1},
\end{align*}
which follows from the combination of \eqref{eq:hook} and \eqref{eq:ac=0:p}.
If $x/n$ is fixed and $n \to \infty$, then 
\begin{align*}
& \log \tbinom{n}{x} 
 =n h(\tfrac{x}{n})+\tfrac{1}{2}\log \tfrac{n}{2\pi x(n-x)}+o(1)
\\
 =& n h(\tfrac{x}{n})-\tfrac{1}{2}\log n
     -\tfrac{1}{2}\log \bigl(\tfrac{2\pi x}{n}(1-\tfrac{x}{n})\bigr)+o(1),
\end{align*}
where $h(t) := -t \log t - (1-t) \log(1-t)$ is the binary entropy.
Then, a direct computation yields 
\begin{align}\label{eq:ac=0:Sdfn}
\begin{split}
& \log \dim \clV_{(n-x,x)} - \log p(x)\\
=&n h(\tfrac{m}{n})+m h(\tfrac{x}{m})-l h(\tfrac{x}{l})-(n-l-x) h(\tfrac{m-l}{n-l-x}) 
\\
& -\tfrac{1}{2}\log\tfrac{n m}{l(n-l-x)} 
 -\tfrac{1}{2}\log\tfrac{m}{n}(1-\tfrac{m}{n})\tfrac{x}{m}(1-\tfrac{x}{m}) 
 \\
 &+\tfrac{1}{2}\log\tfrac{x}{l}(1-\tfrac{x}{l})\tfrac{m-l}{n-l-x}(1-\tfrac{m-l}{n-l-x})
 +o(1). 
\end{split}
\end{align}

Hereafter we use the variables and ratios 
\begin{align}
 z:=x-n \mu, \quad 
 \xi=\tfrac{m}{n}, \quad \alpha=\tfrac{l}{n}, \quad \beta=\tfrac{m-l}{n}.
\end{align} 
We can then rewrite \eqref{eq:ac=0:Sdfn} as
\begin{align}
&  \log \dim \clV_{(n-x,x)} - \log p(x) \nonumber\\
\begin{split}
=& n h(\xi)+ n \xi h(\tfrac{\mu}{\xi}+\tfrac{z}{n \xi})
  -n \alpha h(\tfrac{\mu}{\alpha}+\tfrac{z}{\alpha n}) \\
&  -n(1-\alpha-\mu-\tfrac{z}{n}) h(\tfrac{\beta}{1-\alpha-\mu-\frac{z}{n}}) \\
&
  -\tfrac{1}{2}\log\tfrac{\xi}{\alpha(1-\alpha-\mu-\frac{z}{n})} \\
&  -\tfrac{1}{2}\log\xi(1-\xi)(\tfrac{\mu}{\alpha}+\tfrac{z}{n \alpha})
                           (1-\tfrac{\mu}{\alpha}-\tfrac{z}{n \alpha}) \\
&  +\tfrac{1}{2}\log(\tfrac{\mu}{\alpha}+\tfrac{z}{n \alpha})
                 (1-\tfrac{\mu}{\alpha}-\tfrac{z}{n \alpha}) \\
&\cdot  \tfrac{\beta}{1-\alpha-\mu-\frac{z}{n}}(1-\tfrac{\beta}{1-\alpha-\mu-\frac{z}{n}})
\end{split}
\nonumber\\
\begin{split}\label{eq:ac=0:vNS} 
=& n h(\xi)+n \xi h(\tfrac{\mu}{\xi}+\tfrac{z}{n \xi})
  -n \alpha h(\tfrac{\mu}{\alpha}+\tfrac{z}{n \alpha})\\
&  -n(1-\alpha-\mu-\tfrac{z}{n}) h(\tfrac{\beta}{1-\alpha-\mu-\frac{z}{n}}) \\
&
 +\tfrac{1}{2}\log\tfrac{\alpha \beta}{\xi^2(1-\xi)}
 +\tfrac{1}{2}\log(1-\tfrac{\beta}{1-\alpha-\mu-\frac{z}{n}}).
\end{split}
\end{align}

\begin{lem}\label{lem:ac=0:C16}
In the limit  $z \to 0$ and $n \to \infty$, we have 
\begin{align}\label{eq:ac=0:v(z)}
\begin{split}
&\log \dim \clV_{(n-x,x)}-\log p(x) = v(z) + O\Bigl(\frac{z^3}{n^2}\Bigr), \\
&v(z) := C_1 n + C_2 + C_3 z + C_5 \frac{z}{n} + 
         C_4 \frac{z^2}{n}   + C_6 \frac{z^2}{n^2},
\end{split}
\end{align}
where the coefficients $C_1$ through $C_4$ are given 
as in Proposition \ref{prp:ProPH}, and the rest are given by 
\begin{align}\label{eq:ac=0:C1-6}
\begin{split}
 C_5 &:=-\frac{1}{2} \frac{\beta}{(1-\alpha-\mu)(1-\alpha-\mu-\beta)}, 
 \\
 C_6 &:= \frac{1}{4} \Bigl(\frac{1}{(1-\alpha-\mu)^2}
                         -\frac{1}{(1-\alpha-\mu-\beta)^2}\Bigr).
\end{split}
\end{align}
The formula actually holds for any smooth function $h(t)$.
\end{lem}

\begin{proof}
By the Taylor expansion of the binary entropy $h$, 
for $\zeta = \xi,\alpha$,
we have
\begin{align*}
&n \xi h(\tfrac{\mu}{\zeta}+\tfrac{z}{n \zeta}) \\
 =& \xi h(\tfrac{\mu}{\zeta})n + z h'(\tfrac{\mu}{\zeta}) + 
   \tfrac{1}{2\zeta}h''(\tfrac{\mu}{\zeta}) \tfrac{z^2}{n} + O(\tfrac{z^3}{n^2}), 
\\
&n(1-\alpha-\mu-\tfrac{z}{n}) h(\tfrac{\beta}{1-\alpha-\mu-\frac{z}{n}})\\
 =&(1-\alpha-\mu-\tfrac{z}{n}) h(\tfrac{\beta}{1-\alpha-\mu}) n \\
&
  - \bigl(h(\tfrac{\beta}{1-\alpha-\mu})
           -\tfrac{\beta}{1-\alpha-\mu}h'(\tfrac{\beta}{1-\alpha-\mu})\bigr)z 
\\
& +\tfrac{1}{2}\tfrac{\beta^2}{(1-\alpha-\mu)^3} 
  h''(\tfrac{\beta}{1-\alpha-\mu})\tfrac{z^2}{n} + O(\tfrac{z^3}{n^2}).
\end{align*}
We also have 
\begin{align*}
 & \log(1-\tfrac{\beta}{1-\alpha-\mu-\frac{z}{n}}) \\
 = 
& \log(1-\tfrac{\beta}{1-\alpha-\mu}) - 
  \tfrac{\beta}{(1-\alpha-\mu)(1-\alpha-\mu-\beta)}\tfrac{z}{n} \\
&+\tfrac{1}{2}\bigl(\tfrac{1}{(1-\alpha-\mu)^2}-\tfrac{1}{(1-\alpha-\mu-\beta)^2}\bigr)
  \tfrac{z^2}{n^2} + O(\tfrac{z^3}{n^3}).
\end{align*}
Replacing the terms in \eqref{eq:ac=0:vNS} by these expressions,
we obtain \eqref{eq:ac=0:C1-6}.
\end{proof}

Now we can show Proposition \ref{prp:ProPH}.

\begin{proof}[{Proof of Proposition \ref{prp:ProPH}}]
Let us denote $p_n(x) := p(x \midd n,m,k,l)$.
By Lemma \ref{lem:ac=0:C16}, for $\ep>0$,
there exists a constant $C_7$ such that
any $z \in \bigl( (\mu -\ep)n, (\mu+\ep)n \bigr)$ satisfies
\begin{align}\label{eq:ac=0:NN1}
 \abs{\bigl(\log \dim \clV_{(n-x,x)}-\log p_n(x)\bigr)-v(z)} 
 < C_7 \frac{\abs{z}^3}{n^2}.
\end{align}
Since the values $\max_x(\log \dim \clV_{(n-x,x)}-\log p_n(x))$ and $\max_x v(x-n\mu)$
behave linearly for $n$
and Proposition \ref{prp:5.6.2} guarantees that
the probability $\sum_{x: \, \abs{x-n \mu} > \ep} p_n(x)$
goes to zero exponentially,
we have
\begin{align}
0=&\lim_{n \to \infty} \Bigl( \sum_{x: \, \abs{x-n \mu} > \ep}
 \bigl( \log \dim \clV_{(n-x,x)}\nonumber\\
 &-\log p_n(x) \bigr) p_n(x)\Bigr) , 
 \label{eq:ac=0:VB2} \\
0=&\lim_{n \to \infty} \Bigl(\sum_{x: \, \abs{x-n \mu} > \ep}
 v(x-n\mu) p_n(x) \Bigr) .
 \label{eq:ac=0:VB1}
\end{align}
Since \eqref{eq:ac=0:NJ1}, \eqref{eq:ac=0:NJ2} and \eqref{eq:ac=0:v(z)} imply
\begin{align*}
 0=&\lim_{n\to \infty} \Bigl(\bbE[v(X-n\mu)]\\
 &-(n C_1+ C_2+C_3\phi+C_4\sigma^2)\Bigr).
\end{align*}
\eqref{eq:ac=0:VB1} implies 
\begin{align}
0= &
\lim_{n \to \infty} \Bigl( 
  \Bigl(\sum_{x: \, \abs{x-n \mu} \le \ep} v(x-n\mu) p_n(x)\Bigr)\nonumber \\
  &  -(n C_1+ C_2+C_3\phi+C_4\sigma^2)\Bigr) .
\label{eq:ac=0:VB3}\end{align}
Moreover, we have 
\begin{align}\label{eq:ac=0:VB4}
\begin{split}
&\sum_{x: \, \abs{x-n \mu} \le \ep}
 \Bigl( \bigl(\log \dim \clV_{(n-x,x)}-\log p_n(x) \bigr)\\
& - v(x-n\mu) \Bigr) p_n(x) \\
\stackrel{(a)}{\le} &
 \sum_{x: \, \abs{x-n \mu} \le \ep} 
 C_7 \, \tfrac{\abs{x-n \mu}^3}{n^2} \, p_n(x) \\
 \le & C_7 \, \bbE\Big[\abs{\tfrac{X-n\mu}{\sqrt{n}}}^3\Bigr] \, \frac{1}{\sqrt{n}}
 \xrightarrow[ \ n \to \infty \ ]{(b)} 0, 
\end{split}
\end{align}
where $(a)$ and $(b)$ follow from \eqref{eq:ac=0:NN1} and \eqref{NVT}, respectively.
The combination of \eqref{eq:ac=0:VB2}, \eqref{eq:ac=0:VB3}, 
and \eqref{eq:ac=0:VB4} implies \eqref{eq:ac=0:NT1}.
\end{proof}


\section*{Acknowledgement}
The author was supported in part by the National
Natural Science Foundation of China under Grant 62171212.
He is grateful to Professor Akito Hora and Professor Shintaro Yanagida
for helpful discussions on the topic of this paper.
In particular, 
he is thankful for Professor Shintaro Yanagida
to providing Figure \ref{fig:intro:vN} and helping his description of a part of this paper.


\bibliographystyle{quantum}
\bibliography{final}

\end{document}